\documentclass[10pt]{article}
\usepackage{amssymb}
\usepackage{amsfonts}
\usepackage{amsthm}
\usepackage{amsmath}
\usepackage{bm}	
\usepackage{bbm}
\usepackage{hyperref}
\usepackage{cite}
\usepackage{verbatim}

\title{Some improved nonperturbative bounds for Fermionic expansions}
\author{Martin Lohmann \\ \textit{{\small  Dipartimento di Matematica, Universita di Roma Tre}}}
\date{}
\allowdisplaybreaks

\newcommand{\ud}{\,\mathrm{d}}
\newcommand{\bd}[1]{\mathbf{#1}}
\newcommand{\mc}[1]{\mathcal{#1}}
\newcommand{\mf}[1]{\mathfrak{#1}}
\newcommand{\half}[0]{\frac{1}{2}}
\newcommand{\id}[0]{\mathbbm{1}}
\newcommand{\pderi}[1]{\frac{\partial}{\partial #1}}
\newcommand{\ol}[1]{\overline{#1}}
\newcommand{\ul}[1]{\underline{#1}}
\newcommand{\mb}[1]{\mathbb{#1}}
\newcommand{\qqquad}{\qquad\qquad\qquad\qquad\qquad}
\newcommand{\const}{\text{ const }}
\newcommand{\ob}[1]{^{(#1)}}
\newcommand{\sgn}{\text{sgn }}
\newcommand{\Pf}{\text{Pf }}
\newcommand{\Ant}[1]{\text{Ant}_{#1}}

\theoremstyle{plain}
\newtheorem{thm}{Theorem}
\newtheorem{cor}{Corollary}

\theoremstyle{remark}
\newtheorem*{rem}{Remark}

\begin{document}

\maketitle

\begin{abstract}

We reconsider the Gram-Hadamard bound as it is used in constructive quantum field theory and many body physics to prove convergence of Fermionic perturbative expansions. Our approach uses a recursion for the amplitudes of the expansion, discovered in a model problem by Djokic \cite{djokic}. It explains the standard way to bound the expansion from a new point of view, and for some of the amplitudes provides new bounds, which avoid the use of Fourier transform, and are therefore superior to the standard bounds for models like the cold interacting Fermi gas.

\end{abstract}

\section{Introduction}

The mathematical theory of interacting Fermions has seen some impressive successes during the past decades, partly because, in contrast to its Bosonic counterpart, many perturbative arguments from theoretical physics can be made rigorous by using certain natural \emph{convergent} resummations of the Fermionic perturbation series \cite{caianiello1956number,simon1969convergence,gawccdzki1985gross,lesniewski1987effective}. Despite this, for some important models, notably the cold Fermi gas, the bounds proven on these resummations are far off the ones conjectured and used by theoretical physicists, or can brought to match the latter only with substantial effort. The discrepancy between available nonperturbative bounds and perturbative predictions, even though believed by many to be only technical, has since become one of the major obstacles in the field. In this paper, we prove nonperturbative bounds of a new kind, inspired by the solution of Djokic \cite{djokic} to a model problem, stated by Feldman, Kn\"orrer and Trubowitz \cite{feldman1973trubowitz} as an elementary geometric question about permutations. Our bounds do not solve the discrepancy, but illustrate the problem from a new perspective.\\
The convergence of Fermionic perturbation series is a consequence of antisymmetry properties of Fermionic interaction kernels. All current strategies to prove nonperturbative bounds for the series use the Fourier transform (from momentum space to position space) to exploit this antisymmetry. More precisely, in some way or another, the arguments are based on the inequality $\Vert f\Vert_\infty\leq \frac{n^{\frac n2}}{n!} \Vert \hat f\Vert_1 $ that holds for any antisymmetric function $f = f(p),p\in\mb R^n,  $ and its Fourier transform $\hat f = \hat f(x) $ and is a simple consequence of the Gram-Hadamard inequality for determinants. Unfortunately, the Fourier transform of some quantities entering the perturbation series can be ill behaved, and the transform obscures the fundamental principle of momentum conservation, and is therefore avoided in the perturbative arguments of theoretical physics. At the same time, the advantages of Fourier transform over other tools that work more directly in momentum space and provide similar convergence factors, like the improved Poincare inequality $\Vert f\Vert_\infty\leq \frac{n^{\frac n2}}{n!} \Vert \bm \nabla f\Vert_1 $ with $\bm\nabla = \pderi{p_1}\cdots\pderi{p_n} $, have never been systematically described. Many mathematical physicists have therefore searched for nonperturbative bounds that avoid Fourier transform, alas with no success.\\
In this paper, we reconsider the interplay of the antisymmetrization operation and momentum conservation in Fermionic perturbation series, describe how the Fourier transform can be used to manipulate the series to obtain summable, nonperturbative bounds, and describe alternative manipulations that produce nonperturbative bounds which work in momentum space and avoid Fourier transform, but are only summable over a subseries of the whole perturbation series characterized by a simple momentum conservation structure. More precisely, as is done in many proofs of the Gram-Hadamard inequality, we express the antisymmetrization operation in terms of an exterior tensor algebra, in which momentum conservation defines a certain recursion whose solution yields the terms of the perturbation series, similar to Djokic's recursion \cite{djokic} in his model problem. Fourier transform decomposes the antisymmetric tensors produced by the recursion into tensors of rank $1$, which reduces bounds on the recursion to bounds on the wedge product of rank $1$ tensors. Alternative manipulations, of which we propose two examples, decompose into tensors of higher rank, and use bounds on the wedge product of higher rank tensors. This produces nonperturbative bounds for terms of the perturbation series whose momentum conservation structure, as represented by a spanning tree of the Feynman graph of that term, takes on a simple form (in the examples we give, the underlying trees need to have either short or only few branches).\\
The most intensively studied model for which the discrepancy that is the subject of this paper is present, is the nonrelativistic Fermi gas with two body interactions at low temperature. Indeed, the perturbation theory (to any order) of this model has been understood in great detail using the renormalization group \cite{benfatto1990renormalization,feldman1990perturbation,metzner2012functional} in momentum space, and involves, in all space dimensions, a relevant renormalization of the Fermi surface and a marginal flow of the two body interaction. By way of contrast, employing the Fourier transform to position space in a nonperturbative renormalization group study of the same model seems to display a relevant flow of the two body interaction in dimensions $ d\geq 2$\footnote{In $d=1 $, the singularities of the momentum space propagator are points, and perturbative power counting coincides with the nonperturbative one (see \cite{gentile2001renormalization} for a review). The same is true for other models with point singularities, like the Gross-Neveu model. The results of this paper are not needed in such cases.}, and many questions about the model become intractable. Adding subtle momentum space analysis to the nonperturbative argument, in particular sectorization of the Fermi surface in \cite{disertori2000rigorous,benfatto2003low,benfatto2006fermi} for $d=2 $ and strictly positive temperature, and overlapping loops and bounds for ladder graphs in \cite{feldman2004two} for $d=2$, zero temperature, and nonzero magnetic field, some of the perturbative predictions were established rigorously. Similar results in the physical $d=3$ case are not available, in particular because the sectorization arguments are specific to $d=2$. Some partial results for the $d=3$ jellium model that avoid momentum space analysis to a large extend have been obtained, however \cite{magnen1995single,disertori2001interacting,disertori2013parametric}. All these results require great effort compared to the perturbative treatment in momentum space. In this paper, using a relatively simple argument, we explore how momentum space can be used in nonperturbative bounds and show, as an application and motivation, that a certain subseries (characterized by a simple momentum conservation structure) of the perturbation series for the cold Fermi gas, in any dimension, has the same power counting nonperturbatively as it does perturbatively, so that in particular, the two body interaction has a marginal flow (in this subseries).    \\
This paper is organized as follows: In section \ref{notation}, we fix notation for the perturbation series and its resummation we shall use. In section \ref{bounds}, we reduce bounds on the resummed perturbation series to bounds on the antisymmetrization operation, formulate this operation in terms of an exterior algebra, describe how momentum conservation leads to a recursion in this algebra, and discuss the Fourier transform and alternative manipulations as strategies to bound the recursion. Finally, in section \ref{applications}, we give two concrete examples for bounds that can be obtained with our strategy and which avoid Fourier transform to position space, and give a simple minded comparison between the new bounds and the standard Fourier one in the case of a single scale many Fermion system.\\
The focus of this paper is on methodology. We want to emphasize the recursion discovered in \cite{djokic} as a calculational tool for deriving nonperturbative bounds for Fermionic expansions. We do not state theorems or lemmas until section \ref{applications}, where we formulate the two simplest applications of our strategy. The experienced reader might want to look at the bounds of Theorem \ref{thmsparse} and \ref{shortbranches} therein for a quick impression of the scope of our strategy, and at Corollary \ref{corollary} for the application to many Fermion systems. We hope that bounds of this kind will help in the investigation of models where classes of Feynman graphs within their range of application bear physical importance, like the ladder graphs do for cold Fermi gases.

\section{Grassmann Gaussian integrals and Fermionic expansions}\label{notation}

The perturbative expansion of observables of Fermionic quantum fields or many body systems can be conveniently expressed in terms of Grassmann algebras and Grassmann Gaussian integrals. In this language, the algebraic steps of the expansion are expressed in terms of basic mathematical operations, notably the computation of Pfaffians and certain interpolations based on the Taylor expansion. In this section, we quickly discuss the relevant notions. This is standard material, which we include only to introduce notation. Complete references for Grassmann algebras and integrals in many body physics and quantum field theory are \cite{feldman2002fermionic,berezin}, introductions to Fermionic expansions are \cite{brydges1988mayer,abdesselam1998explicit,feldman1998representation,salmhofer2000positivity,gentile2001renormalization,giuleshouches}.

\subsection{Grassmann calculus}

\subsubsection{Grassmann algebras of interactions}

Interactions of Fermionic particles are modeled by elements $W $ of the Grassmann algebra $$\mf W   = \bigoplus_{n\geq 1} \bigwedge^n\mc V $$ of the vector space $\mc V = \mb C^{\mb L} $. We shall take the lattice $\mb L = \epsilon_1 \mb Z/L_1\epsilon_1\mb Z\times \cdots\times \epsilon_D\mb Z/L_D\epsilon_D\mb Z \times\Sigma =: \mb T \times \Sigma $ to be the (discrete) torus times a finite index set, with small, possibly zero, $\epsilon_c $, and large, possibly infinite $L_c\in\mb N $. We denote its elements $\xi = (x,\sigma)\in\mb T\times\Sigma $. If $\psi(\xi),\,\xi\in\mb L $, is the standard basis of $\mc V $, each $W\in\mf W $ can be written uniquely as
\begin{align*}
W &= \sum_{n\geq 1} \sum_{\xi_1,\ldots,\xi_n\in\mb L} w_n(\xi_1,\ldots,\xi_n) \psi(\xi_1)\wedge \cdots\wedge\psi(\xi_n) = \sum_{n\geq 1} W_n,
\end{align*}
with $w_n(\xi_1,\ldots,\xi_n)\in \mb C $ antisymmetric in its arguments. If a total order on $\mb L $ is chosen, the family $\psi(\xi_1)\wedge\cdots\wedge \psi(\xi_n),\xi_1<\cdots<\xi_n $, is a basis of $\bigwedge^n\mc V $. In this basis, $W$ is expressed as
\begin{align*}
W &= \sum_{n\geq 1} \sum_{\xi_1<\cdots<\xi_n} w_{n,o}(\xi_1,\ldots,\xi_n) \psi(\xi_1)\wedge \cdots\wedge\psi(\xi_n)
\end{align*}
with $ w_{n,o}(\xi_1,\ldots,\xi_n) = n!w_n(\xi_1,\ldots,\xi_n)  $. The product of $W,W'\in \mf W $ is
\begin{align*}
WW' &= \sum_{n,n'\geq 1} \sum_{\substack{ \xi_1,\ldots,\xi_n\in\mb L\\ \xi_1',\ldots,\xi_{n'}'\in\mb L }} w_n(\xi_1,\ldots,\xi_n)w'_{n'}(\xi_1',\ldots,\xi_{n'}') \\&\qqquad\times \psi(\xi_1)\wedge\cdots\wedge\psi(\xi_n)\wedge\psi(\xi_1')\wedge\cdots\wedge \psi(\xi_{n'}')\\
&= \sum_{n\geq 1} \sum_{\xi_1,\ldots,\xi_n\in\mb L} [ww']_n(\xi_1,\ldots,\xi_n)  \psi(\xi_1)\wedge\cdots\wedge\psi(\xi_n)
\end{align*}
where
\begin{align*}
[ww']_n(\xi_1,\ldots,\xi_n) &= \sum_{n_1=1}^{ n-1}  \Ant{} \big[w_{n_1}(\xi_1,\ldots,\xi_{n_1}) w_{n-n_1}'(\xi_{n_1+1},\ldots,\xi_n)\big]
\end{align*}
with
\begin{align*}
\Ant{} f(\xi_1,\ldots,\xi_n) &= \frac 1{n!} \sum_{\pi\in\mc S_n} \sgn\pi f(\xi_{\pi(1)},\ldots,\xi_{\pi(n)}).
\end{align*}
The partial antisymmetry of the expression $w_{n_1}(\xi_1,\ldots,\xi_{n_1}) w_{n-n_1}'(\xi_{n_1+1},\ldots,\xi_n) $ would allow to restrict the sum over permutations in Ant to shuffle permutations (i.e. permutations that are order preserving on $1,\ldots, n_1 $ and $n_1+1,\ldots,n $). Similar formulas and considerations apply to higher products $W_1\cdots W_k $.\\[5pt]
If $\epsilon_c>0 $ and $L_c<\infty $, $c=1,\ldots,D $, the above sums are all finite. The limiting cases $\epsilon_c\to0, L_c\to\infty $ can be treated as in Appendix A in \cite{feldman2002fermionic}. The technicalities due to interchanges of sums/integrals are not complicated and we use the finite dimensional notation $\sum_x $ even in the case where really $\int \ud x $ is meant.

\subsubsection{The Grassmann Gaussian integral}

The Grassmann Gaussian integral encodes the properties of the noninteracting Fermion system that is perturbed by the interaction $W\in\mf W $. Given an antisymmetric kernel $C:\mb L\times\mb L\to \mb C $, it is defined to be the linear functional on $\mf W$ determined by
\begin{align*}
\int \psi(\xi_1)\wedge\cdots\wedge\psi(\xi_n)\ud\mu_C(\psi) &= \Pf \big(C(\xi_m,\xi_{m'})\big)_{m,m'\in\ul n},
\end{align*}
where the Pfaffian is $\text{Pf }\big(C(\xi_m,\xi_{m'})\big)_{m,m'\in\ul n} = 0  $ if $n $ is odd and
\begin{align*}
\Pf\big(C(\xi_m,\xi_{m'})\big)_{m,m'\in\ul{2n'}} &= \frac1{2^{n'}n'!} \sum_{\pi\in\mc S_{2n'}} \sgn\pi \prod_{m=1}^{n'} C(\xi_{\pi(2m-1)},\xi_{\pi(2m)} ).
\end{align*}
The Pfaffian of the empty matrix is $1$. $C$ is referred to as the covariance of the Grassmann Gaussian integral. We have written here $\ul n = \{1,\ldots,n\} $ for any $n\in\mb N $. In particular, for $W\in\mf W $, we have
\begin{align*}
\int W\ud\mu_C(\psi) &= \sum_{n\geq 1} \frac{(2n)!}{2^nn!}\sum_{\xi_1,\ldots,\xi_{2n}\in\mb L}  w_{2n}(\xi_1,\ldots,\xi_{2n}) \prod_{m=1}^n C(\xi_{2m-1},\xi_{2m})
\end{align*}

\subsubsection{Translation invariance and momentum space}

The easiest situation in which the correct volume factors for thermodynamic quantities can be extracted is for translation invariant systems. Such systems are characterized by an interaction of the form 
\begin{align*}
W &= \sum_{n\geq 1} \sum_{\substack{\xi_1,\ldots,\xi_n\in\mb L \\ x\in \mb T}} w_n(\xi_1,\ldots,\xi_n)\psi(x+\xi_1)\wedge\cdots\wedge \psi(x+\xi_n)
\end{align*}
and by a covariance satisfying $C(\xi,\xi') = C(x+\xi,x+\xi') $. Here, for $x\in\mb T $ and $\xi = (x',\sigma)\in\mb L = \mb T\times\Sigma, $ $x+\xi = (x+x',\sigma)\in\mb L $. From now on, we restrict ourselves to translation invariant interactions and covariances.\\
By Fourier transform, we have
\begin{align*}
C(\xi,\xi') &= \vert \mb T\vert^{-1}\sum_{p\in\mb T^*} \hat C_{\sigma,\sigma'}(p) e^{ip(x-x')}\\
&=:  \int_{\mb T^*}\ud p\; \hat C_{\sigma,\sigma'}(p)  e^{ip(x-x')},
\end{align*}
where $\mb T^* = \frac{2\pi}{\epsilon_1L_1}\mb Z/\frac{2\pi}{\epsilon_1} \mb Z \times \cdots \times  \frac{2\pi}{\epsilon_DL_D}\mb Z/\frac{2\pi}{\epsilon_D} \mb Z  $ and
\begin{align*}
\hat C_{\sigma,\sigma'}(p) &= \sum_{x\in\mb T}C\big((0,\sigma),(x,\sigma')\big)e^{ipx} =:\sum_{x\in\mb T}C_{\sigma,\sigma'}\big(x\big)e^{ipx}  .
\end{align*} 
We denote $\mb L^* = \mb T^*\times \Sigma $ and write for its elements $\lambda = (p,\sigma)\in\mb L^* $. We introduce the weighted sums $$\int_{\mb L^*}\ud\lambda =  \sum_{\sigma\in\Sigma}\int_{\mb T^*}\ud p \qquad\text{and}\qquad \int_{\mb L^*\times\Sigma}\ud\lambda = \sum_{\sigma,\sigma'\in\Sigma}\int_{\mb T^*}\ud p $$ and, for $\lambda = (p,\sigma,\sigma')\in \mb L^*\times\Sigma $, the inversion $-\lambda = (-p,\sigma',\sigma) $. The Grassmann Gaussian integral then becomes
\begin{align*}
\int W\ud\mu_C(\psi) &= \vert\mb T\vert \cdot \sum_{n\geq 1} \frac{(2n)!}{2^nn!}\int_{\mb L^*\times\Sigma}\ud\lambda_1\cdots\ud\lambda_n\,  \hat w_{2n}(\lambda_1,-\lambda_1,\ldots,\lambda_n,-\lambda_{n}) \prod_{m=1}^n \hat C(\lambda_m)
\end{align*}
with $(\lambda_m=(p_m,\sigma_m)$ or $=(p_m,\sigma_m,\sigma_m'))$
\begin{align*}
\hat w_n(\lambda_1,\ldots,\lambda_n) &= \sum_{x_1,\ldots,x_n\in\mb T} w_n\big((x_1,\sigma_1),\ldots,(x_n,\sigma_n)\big)e^{ip_1x_1+\cdots+ip_nx_n}.
\end{align*}
We have
\begin{align*}
W_n &= \int_{\mb L^* }\ud\lambda_1\cdots\ud\lambda_n \,\vert \mb T\vert \cdot \delta_{p_1+\cdots+p_n=0}\cdot  \hat w_n(\lambda_1,\ldots,\lambda_n) \,\hat\psi(\lambda_1)\wedge\cdots\wedge \hat \psi(\lambda_n)
\end{align*}
with
\begin{align*}
\hat \psi(\lambda) &= \sum_{x\in\mb T} \psi(x,\sigma) e^{-ipx}.
\end{align*}

\subsection{Fermionic expansions}

We now define some observables of quantum statistical mechanics in terms of Grassmann Gaussian integrals. In theoretical physics, recursive properties of the Pfaffian are used to generate an expansion of these integrals into a perturbative series over Feynman graphs. As is well known, this series is not absolutely convergent. However, certain resummations of the perturbation series, corresponding to a more selective application of the recursive properties, are convergent. The bounds of the next section should work for many versions of such expansion schemes, but for definiteness, we will explicitly introduce one such scheme,  due to Brydges and Wright \cite{brydges1988mayer}, in this section. We shall not need many details of this scheme in the remaining paper.

\subsubsection{Observables in the Grassmann formalism}\label{observables}

In theoretical physics, the properties of a system of Fermions are derived from a Hamiltonian that is expressed in terms of creation and annihilation operators which satisfy the canonical anticommutation relations. Using the method of coherent states \cite{berezin,salmhofer2009clustering,salmhofer2013renormalization}, computations with creation and annihilation operators can be expressed in terms of Grassmann Gaussian integrals. For simplicity, we focus here on the free energy, which in this formalism becomes
\begin{align}
\Omega_C(W) &= \log  \int e^{W(\psi)} \ud\mu_C(\psi) . \label{rgmap}
\end{align}
Here, the interaction $W\in\mf W $ is determined by the interacting part of the Hamiltonian by simple replacements of the creation and annihilation operators with generators $\psi(x) $, and the kernel $C$ is determined by the noninteracting part of the Hamiltonian. \\
The archetypical case of $d$ dimensional free nonrelativistic Fermions at zero temperature, perturbed by an interaction through a two body potential $v( x), x\in\mb R^d $, gives rise to
\begin{align*}
C\big((x^0, x;\sigma;\kappa),(x_0', x';\sigma';\kappa')\big) &= \Ant{}\;\delta_{\sigma,\sigma'} \delta_{\kappa,1}\delta_{\kappa',0}  \int_{\mb R^{d+1}} \frac{\ud p^0\ud^{d }p}{(2\pi)^{d+1}}\frac{e^{ip(x-x')} }{ip^0-\frac{ p^2}{2m}+\mu}\\
(x^0,x;\sigma;\kappa)&\in \mb L= \mb R\times \mb R^{d}\times \{\uparrow,\downarrow\}\times \{0,1\}
\end{align*}
and
\begin{align*}
W &= -\half \sum_{\sigma,\sigma'\in\{\uparrow,\downarrow\}} \int_{\mb R}\ud x^0 \int_{\mb R^{2d}} \ud   x\ud  x'  v(  x-  x')\\&\qquad\qquad\qquad\times\psi(x^0, x;\sigma;1)\wedge\psi(x^0, x';\sigma';1)\wedge\psi(x^0, x;\sigma;0)\wedge\psi(x^0,  x';\sigma';0) .
\end{align*}
Clearly, the integral for $C $ is singular and has bad decay properties. It is an implementation of the renormalization group strategy to replace its integrand by
\begin{align*}
\frac{e^{ip(x-x')} }{ip^0-\frac{ p^2}{2m}+\mu} \longrightarrow \frac{\chi_j(p^0, p) e^{ip(x-x')} }{ip^0-\frac{ p^2}{2m}+\mu} ,
\end{align*}
where $\chi_j\in C^\infty_c(\mb R^{d+1}) $ is supported away from the singularities of the original integrand and $\sum_j \chi_j = 1 $ almost everywhere. Adding an external field, the map $W\to\Omega_C(W)  $ fulfills a ``semigroup property'' $\Omega_{\sum_jC_j} = \circ_j\Omega_{C_j}, $ which reduces the problem to iterating the more regular maps $\Omega_{C_j} $. The bounds of the next section are intended for the analysis of such a single scale map, as expanded in the way we will now describe. We drop the subscript from $\Omega_C $.

\subsubsection{Fermionic expansions}\label{fermexp}

It follows easily from the definition of Grassmann Gaussian integration that
\begin{align*}
\int  W(\psi)^{\wedge n} \ud\mu_C(\psi) &= \int \bigwedge_{m=1}^n W(\psi^m) \ud\mu_{C\otimes \bd 1}(\psi^1,\ldots,\psi^n)
\end{align*}
with $\bd 1(m,m') \equiv 1 ,m,m'=1,\ldots,n$. Expanding the exponential in the definition of $\Omega(W) $, we obtain
\begin{align*}
\Omega(W) &= \log 1+\sum_{n\geq 1} \frac1{n!} \int \bigwedge_{m=1}^n W(\psi^m) \ud\mu_{C\otimes \bd 1}(\psi^1,\ldots,\psi^n).
\end{align*}
For $s(\{m,m'\})\in [0,1], m\neq m'\in\ul n $, we denote by $s $ the symmetric $n\times n $ matrix with $s_{m,m'} = s(\{m,m'\}) $ and $s_{m,m}=1 $. Then $\bd 1 = s\vert_{s(\{m,m'\})=1\,\forall\,m,m'} $. By the BKAR Taylor interpolation formula \cite{abdesselam1995trees}, we have
\begin{align*}
\int \bigwedge_{m=1}^n W(\psi^m) \ud\mu_{C\otimes \bd 1}(\psi^1,\ldots,\psi^n) &= \sum_{k\geq 1}\sum_{\substack{m_1+\cdots+ m_k= n \\ m_j\geq 1}} \frac{n!}{m_1!\cdots,m_k!} \prod_{j=1}^k \sum_{T\text{ tree on }\ul m_j} \\& \times\int_0^1 \prod_{\ell\in T} \ud s_\ell  \int  \Delta^T\bigwedge_{l=1}^{m_j} W(\psi^l) \ud\mu_{C\otimes s^T}(\psi^1,\ldots,\psi^{m_j})\\
\Delta^T &= \prod_{\{l,l'\}\in T}\sum_{\xi,\xi'\in\mb L}   C(\xi,\xi') \pderi{\psi^l(\xi)}\pderi{\psi^{l'}(\xi')} .
\end{align*}
Derivatives anticommute with each other and with the fields. Above, $s^T $ is the symmetric $m_j\times m_j $ matrix corresponding to $$s^T(\{l,l'\}) = \min \{s_\ell,\ell \text{ on the }T \text{ path between }l,l'\}. $$ It is easy to see that $s^T $ is a convex combination of positive semidefinite matrices, such that $s^T = (a^T)^*a^T $ for some symmetric positive semidefinite $m_j\times m_j $ matrix with $\sum_{l'} \vert a^T_{l,l'}\vert^2 = s^T(l,l)=1 $. The sum over $n$ and the logarithm can now be performed and we obtain
\begin{align}\label{treeexp}
\Omega(W) &= \sum_{m\geq 1} \frac1{m!} \sum_{T\text{ tree on }\ul m} \int_0^1\prod_{\ell\in T} \ud s_\ell  \int  \Delta^T\bigwedge_{l=1}^{m} W(\psi^l) \ud\mu_{C\otimes s^T}(\psi^1,\ldots,\psi^{m}).
\end{align}
This expansion corresponds to a resummation of the full Feynman graph expansion that combines all graphs with the same spanning tree. The interpolating integrals compensate for the fact that a graph may have several spanning trees. Many other resummations, such as those of \cite{giuleshouches,abdesselam1998explicit,disertori2000continuous}, are based on the same idea. The structure $s^T = (a^T)^*a^T $ of the interpolating matrix is very important for the bounds on these expansions. A different expansion is the one of \cite{feldman1998representation}, which builds the spanning tree inductively and needs no interpolation of the covariance $C$.

\section{Bounds on Fermionic expansions}\label{bounds}

We will now describe how to obtain estimates on the expansion for $\Omega(W) $ of the last section. We will ignore the interpolation and set $s^T = \bd 1 $. The minor modifications for general $s^T $ will be commented on at the end of the discussion.\\
The term in the expansion corresponding to a fixed tree $T $ on $\ul m $ is (at $s^T=\bd 1 $)
\begin{align}\label{eqamplitude}
\mc A(T) &=  \sum_{n_1,\ldots,n_m\geq 1} \varpi(T,n_1,\ldots,n_m)\int \mc A(T;n_1,\ldots,n_m;\psi)\ud\mu_C(\psi)\\
\mc A(T;n_1,\ldots,n_m;\psi) &=  \sum_{\substack{\xi_{(l,\ell)}\in\mb L \\ \ell\in T,l\in\ell }}\prod_{\ell=\{l,l'\}\in T} C(\xi_{(l,\ell)},\xi_{(l',\ell)}) \bigwedge_{l=1}^{m} \prod_{\ell\ni l}\pderi{\psi(\xi_{(l,\ell)})} W_{n_l}(\psi) \nonumber
\end{align}
The sign $ \varpi(T,n_1,\ldots,n_m)$ depends on the way we conventionally order the products $\prod_{\ell\ni l} $ and the lines $\ell=\{l,l'\}\in T $. We will discuss these ordering issues in more detail soon, but we will bound the integral of each $\mc A(T;n_1,\ldots,n_m;\psi)$ individually and are not interested here in cancellations between different trees $T$ or values of $n_l$, so an explicit expression for $\varpi $ is irrelevant to our purpose. Indeed, set $\bd n = (n_1,\ldots,n_m) $ and suppose that
\begin{align*}
\mc A(T;\bd n;\psi) &= \sum_{n\geq 1} \sum_{\substack{\xi_1,\ldots,\xi_n\in\mb L \\  x\in\mb T }} \mc A_n(T;\bd n;\xi_1,\ldots,\xi_n)\psi(x+\xi_1)\wedge\cdots\wedge \psi(x+\xi_n).
\end{align*}
with $ \mc A_n(T;\bd n;\xi_1,\ldots,\xi_n)$ antisymmetric. Then
\begin{align}\nonumber
\left\vert\int \mc A(T;\bd n;\psi)\ud\mu_C(\psi)\right\vert &\leq \vert\mb T\vert\cdot  \sum_{n\geq 1} \frac{(2n)!}{2^nn!} \int_{\mb L^*\times\Sigma} \prod_{m=1}^n \ud\lambda_m\vert\hat C(\lambda_m)\vert \\&\qqquad\times \big\vert\hat{\mc A}_{2n}(T;\bd n;\lambda_1,-\lambda_1,\ldots,\lambda_n,-\lambda_n)\big\vert\nonumber\\
&\leq \vert\mb T\vert\cdot\sum_{n\geq 1} \frac{(2n)!}{2^nn!} \Vert\hat C\Vert_1^n\cdot  \Vert\hat{\mc A}_{2n}(T;\bd n)\Vert_\infty \label{loopbound}
\end{align}
with $ \Vert\hat C\Vert_1 = \int_{\mb L^*\times\Sigma} \ud\lambda\vert\hat C(\lambda)\vert $, and so
\begin{align*}
\vert \Omega(W)\vert &\leq \vert\mb T\vert\cdot\sum_{m\geq 1} \frac1{m!} \sum_{T\text{ tree on }\ul m} \sum_{n\geq 1} \frac{(2n)!}{2^nn!} \Vert \hat C\Vert_1^n \sum_{\bd n} \Vert \hat{\mc A}_{2n}(T;\bd n)\Vert_\infty
\end{align*}
In the next subsection, we describe how the kernel $\hat{\mc A}_n(T;\bd n;\lambda_1,\ldots,\lambda_n) $ can be written as the antisymmetrization of a simple expression $\hat{\mc A}'_n(T;\bd n;\lambda_1,\ldots,\lambda_n) $ in $\hat C $ and $\hat w $. $\hat{\mc A}_n'(T;\bd n) $ fulfills the (perturbative) bound
\begin{align}\label{perestimate}
\Vert\hat{\mc A}'_n(T;\bd n) \Vert_\infty &\leq \delta_{n=n(\bd n,m)}   \cdot \big(\text{const } \Vert \hat C\Vert_\infty\big)^{m-1} \prod_{l=1}^m d^T(l)!2^{n_l}\Vert \hat w_{n_l}\Vert_\infty 
\end{align}
with $n(\bd n,m) = \sum_l n_l-2(m-1) $ and $d^T(l) $ the degree of $l$ in $T$. Had we used the trivial bound $\Vert \hat{\mc A}_n(T;\bd n)\Vert_\infty  = \Vert \Ant{}\hat{\mc A}'_n(T;\bd n)\Vert_\infty \leq \Vert \hat{\mc A}'_n(T;\bd n)\Vert_\infty $, this would give
\begin{align*}
\vert\Omega(W)\vert &\leq  \vert \mb T\vert\cdot \sum_{m\geq 1} \frac1{m!}  \big(\text{const } \Vert \hat C\Vert_\infty\big)^{m-1} \Vert \hat C \Vert_1\sum_{T\text{ tree on }\ul m} \prod_{l=1}^m d^T(l)! \\&\qqquad\times  \sum_{n_1,\ldots,n_m\geq 1}  n(\bd n,m)^{\half n(\bd n,m)} \prod_{l=1}^m\Vert \hat C\Vert_1^{\frac{n_l-2}2} 2^{n_l}\Vert \hat w_{n_l}\Vert_\infty\\
&\leq \vert \mb T\vert\cdot \sum_{m\geq 1} 8^m   \big(\text{const } \Vert \hat C\Vert_\infty\big)^{m-1} \Vert \hat C \Vert_1 \sum_{n_1,\ldots,n_m\geq 1}  n(\bd n,m)^{\half n(\bd n,m)} \prod_{l=1}^m\Vert \hat C\Vert_1^{\frac{n_l-2}2} 2^{n_l}\Vert \hat w_{n_l}\Vert_\infty
\end{align*}
It is easily seen that these sums cannot be controlled unless $\hat w_n = 0 $ for all $n\geq 3 $. In this section, we show how to use the cancellations in the antisymmetrization operation to obtain nonperturbative bounds of the form 
\begin{align}\label{improbound}
\Vert\hat{\mc A}_n(T;\bd n) \Vert_\infty \leq\delta_{n=n(\bd n,m)} \cdot  n(\bd n,m)!^{-\half} \Vert \hat C\Vert^{m-1} \prod_{l=1}^m d^T(l)!2^{n_l}\Vert \hat w_{n_l}\Vert, 
\end{align}  with different norms on $\hat C $ and $\hat w_n $ than the supremum norm. It is clear that this bound, together with a generic smallness assumption on $\Vert\hat w_n\Vert $, implies a summable bound on $\vert \Omega(W)\vert $, but unfortunately, conceivable norms $\Vert \hat C\Vert $ on the right hand side perform poorly in some important applications. In particular, in the case of the cold Fermi gas in $d\geq 2 $,  norms involving the Fourier transform  (such as $\Vert C\Vert_1 $, for which the standard approach proves (\ref{improbound})) are much larger than the norm $\Vert \hat C\Vert_\infty $ of the perturbative bound (\ref{perestimate}). The motivation of this paper is to shed some dim light on this trade off between acceptable norms and factorials.

\subsection{Antisymmetrization of the kernels of the tree expansion}

Performing the derivatives in the definition (\ref{eqamplitude}) of $\mc A(T;\bd n;\psi) $, we get
\begin{align*}
\mc A(T;\bd n;\psi) &= \sum_{\substack{\xi^l_1,\ldots,\xi^l_{n_l-d^T(l)} \in\mb L \\ x^l\in\mb T  \\ l=1,\ldots,m }} \Bigg[ \sum_{\substack{\xi_{(l,\ell)}\in\mb L \\ \ell\in T,l\in\ell }}\prod_{\ell=\{l,l'\}\in T} C(x^l+\xi_{(l,\ell)},x^{l'}+\xi_{(l',\ell)})\\&\qqquad\times  \prod_{l=1}^m \frac{n_l!}{(n_l-d^T(l))!} w_{n_l}\big((\xi_{(l,\ell)})_{\ell\ni l},\xi^l_1,\ldots,\xi^l_{n_l-d^T(l)}\big)\Bigg] \\&\qqquad\qquad\times \bigwedge_{l=1}^m \psi(x^l+\xi^l_1)\wedge\cdots\wedge \psi(x^l+\xi^l_{n_l-d^T(l)})
\\&= \vert \mb T\vert\cdot \int_{\mb L^*}\prod_{l=1}^m\prod_{k=1}^{n_l-d^T(l)}\ud\lambda^l_k\; \Bigg[ \sum_{\substack{\lambda_{\ell}\in\mb L^*\times\Sigma \\ \ell\in T  }}\prod_{\ell \in T} \hat C(\lambda_\ell) \prod_{l=1}^m \frac{n_l!}{(n_l-d^T(l))!}\\&\qquad\times  \delta\Big(\big(\varsigma_{(\ell,l)}\lambda_\ell\big)_{\ell\ni l} , \lambda_1^l,\ldots,\lambda^l_{n_l-d^T(l)}  \Big)\cdot   \hat w_{n_l}\big((\varsigma_{(\ell,l)} \lambda_{\ell})_{\ell\ni l},\lambda^l_1,\ldots,\lambda^l_{n_l-d^T(l)}\big)\Bigg] \\&\qqquad\qqquad\times \bigwedge_{l=1}^m \hat\psi(\lambda^l_1)\wedge\cdots\wedge \hat\psi(\lambda^l_{n_l-d^T(l)})\\
&=:  \vert \mb T\vert\cdot \int_{\mb L^*}\prod_{l=1}^m\prod_{k=1}^{n_l-d^T(l)}\ud\lambda^l_k\;   \hat {\mc A}_{n(\bd n,m)}'(T;\bd n;\lambda^1_1,\ldots,\lambda^m_{n_m-d^T(m)}) \\&\qqquad\qqquad\times \bigwedge_{l=1}^m \hat\psi(\lambda^l_1)\wedge\cdots\wedge \hat\psi(\lambda^l_{n_l-d^T(l)})
\end{align*}
We have used the notation $\delta((p_1,\cdot), \ldots , (p_n,\cdot)) = \delta_{p_1+\cdots+p_n,0} $. The sequences $(\xi_{(l,\ell)})_{\ell\ni l} $ and $(\varsigma_{(\ell,l)} \lambda_{\ell})_{\ell\ni l} $ are ordered in the opposite order that was conventionally chosen for $ \prod_{\ell\ni l}\pderi{\psi(\xi_{(l,\ell)})}  $. The sign $\varsigma_{(\ell,l)} $ is $+1$ if $\ell=\{l,l'\} $ is conventionally thought of to be oriented towards $l $ and $-1 $ otherwise. The choice of these conventions produces the overall sign $\varpi(T,\bd n) $ in (\ref{eqamplitude}). Its value and the ordering of the sequences $(\xi_{(l,\ell)})_{\ell\ni l}, (\varsigma_{(\ell,l)} \lambda_{\ell})_{\ell\ni l} $ is irrelevant to our discussion, and the orientation of the lines $\ell $ will be discussed shortly. $\hat {\mc A}_n'(T;\bd n;\lambda^1_1,\ldots,\lambda^m_{n_m-d^T(m)})$ is not explicitly written as the coefficient of a translation invariant interaction, and is in general not antisymmetric under exchange of $\lambda^l_k,\lambda^{l'}_{k'} $ unless $l=l' $.

\subsubsection{Momentum conservation and rooted trees}

The expression for $\hat{\mc A}'_n(T;\bd n) $ contains one unnormalized sum over momentum space $\mb L^*\times\Sigma $ for each edge in $T$. These sums are constrained by a momentum conservation delta function $\delta(\lambda_1,\ldots,\lambda_{n_l}) $ for each vertex $l$ of $T$. Clearly, these constraints can only be satisfied if the sum of all external momenta equals zero, in which case they admit a unique solution. This solution can be constructed inductively by using the partial order on the tree associated to picking arbitrarily a root vertex.\\
More precisely, let $r$ be an arbitrary choice for a root of $T$, and for $\ell\in T $ denote by $\mf o(\ell) $ the set of vertices of $T$ whose path to $r $ includes $\ell$. Let $\varsigma_\ell $ equal $+1$ if $\ell $ is conventionally thought of to be oriented towards $r$ and $-1$ otherwise. Then
\begin{align*}
\prod_{l=1}^m \delta\Big(\big(\varsigma_{(\ell,l)}\lambda_\ell\big)_{\ell\ni l} , \lambda_1^l,\ldots,\lambda^l_{n_l-d^T(l)}  \Big) &= \delta\big(\lambda_1^1,\ldots,\lambda^m_{n_m-d^T(m)}\big) \prod_{\ell\in T}\delta_{p_\ell, \varsigma_\ell\sum_{l\in\mf o(\ell)}\sum_{k=1}^{n_l-d^T(l)}p^l_k  }
\end{align*}
where we wrote $\lambda^l_k = (p^l_k,\sigma^l_k) $ and $\lambda_\ell=(p_\ell,\sigma_\ell,\sigma_\ell') $. With $\bm\lambda = (\lambda^1_1,\ldots,\lambda^m_{n_m-d^T(m)}) $ and $$\lambda_\ell(\bm\lambda) = \Big( \varsigma_\ell\sum_{l\in\mf o(\ell)}\sum_{k=1}^{n_l-d^T(l)}p^l_k ,\sigma_\ell,\sigma_\ell'\Big) , $$ we then have
\begin{align}\label{exprap}
\hat {\mc A}_{n(\bd n,m)}'(T;\bd n;\bm\lambda) &= \delta(\bm\lambda)\sum_{\substack{\sigma_\ell,\sigma_\ell'\in\Sigma  \\ \ell\in T}} \prod_{\ell \in T} \hat C(\lambda_\ell(\bm\lambda)) \prod_{l=1}^m \frac{n_l!}{(n_l-d^T(l))!} \\&\qqquad\times    \hat w_{n_l}\big((\varsigma_{(\ell,l)} \lambda_{\ell}(\bm\lambda))_{\ell\ni l},\lambda^l_1,\ldots,\lambda^l_{n_l-d^T(l)}\big)\nonumber
\end{align}
We denote with $\Pi(l) $ the first vertex on the path from $l $ to $r$ (predecessor). We will from now on choose sign conventions such that all lines $\ell = \{l,\Pi(l)\} $ are oriented towards $\Pi(l) $, and that in the list $((\varsigma_{(\ell,l)} \lambda_{\ell}(\bm\lambda))_{\ell\ni l} $, $\ell = \{l,\Pi(l)\} $ appears first. Then $\varsigma_\ell=1 $ for all $\ell $ and 
\begin{align*}
\hat {\mc A}_{n(\bd n,m)}'(T;\bd n;\bm\lambda) &= \delta(\bm\lambda)\sum_{\substack{\sigma_\ell,\sigma_\ell'\in\Sigma  \\ \ell\in T}} \prod_{\ell \in T} \hat C(\lambda_\ell(\bm\lambda)) \prod_{l=1}^m \frac{n_l!}{(n_l-d^T(l))!} \\&\qquad\qquad\times    \hat w_{n_l}\big(-\lambda_{\{l,\Pi(l)\}}(\bm\lambda),(\lambda_{\{l,l'\}} (\bm\lambda))_{l'\in \Pi^{-1}(l)},\lambda^l_1,\ldots,\lambda^l_{n_l-d^T(l)}\big)
\end{align*}

\subsubsection{Recursive structure of the antisymmetrization}\label{algrec}

Let $\mc I = \{(1,1),\ldots,(m,n_m-d^T(m))\} $ be the set of indices of momenta $\lambda_\iota,\iota\in\mc I, $ that are passed as arguments to $ \hat {\mc A}_{n(\bd n,m)}'(T;\bd n) $. We write $\iota<\iota' $ for the order shown and abbreviate $\bm\lambda $ for the sequence $\lambda_{\iota} ,\iota\in\mc I $ in that order. $\lambda_\iota= (p_\iota,\sigma_\iota) $ is fixed in this section. Then
\begin{align*}
\hat {\mc A}_{n(\bd n,m)}(T;\bd n;\bm\lambda) &= \frac1{n(\bd n,m)!} \sum_{\pi\in \mc S(\mc I)} \sgn \pi \; \hat {\mc A}_{n(\bd n,m)}'(T;\bd n;\pi\bm\lambda)
\end{align*}
with $(\pi\bm\lambda)_\iota = \lambda_{\pi(\iota)} $. The exterior algebra $\mf G = \oplus_{n\geq 0} \big(\mb C^{\mc I}\big)^{\wedge n} $ has the standard basis $e_{\iota_1}\wedge\cdots\wedge e_{\iota_n},\iota_m\in\mc I,\iota_1<\cdots<\iota_n $. We define the linear functional $\int:\mf G\to\mb C $ through
\begin{align*}
\int e_{\iota_1}\wedge\cdots\wedge e_{\iota_n} &= \left\{ \begin{array}{ll}   1 &\text{if }n=n(\bd n,m),\iota_1<\cdots<\iota_n  \\ 0 & \text{otherwise}  \end{array}\right.
\end{align*}
In particular,
\begin{align*}
\sgn\pi &= \int e_{\pi((1,1))}\wedge\cdots\wedge e_{\pi((m,n_m-d^T(m)))},
\end{align*}
and therefore
\begin{align}\nonumber
\hat {\mc A}_{n(\bd n,m)}(T;\bd n;\bm\lambda) &= \frac1{n(\bd n,m)!}  \int \alpha(T;\bd n;\bm\lambda)\\
\alpha(T;\bd n;\bm\lambda)&=  \sum_{\iota_{(1,1)},\ldots,\iota_{(m,n_m-d^T(m))}\in\mc I}\hat {\mc A}_{n(\bd n,m)}'(T;\bd n;\lambda_{\iota_{(1,1)}},\ldots,\lambda_{ \iota_{(m,n_m-d^T(m))} }) \label{basictopform}\\&\qqquad\qquad\qquad\qquad \times e_{\iota_{(1,1)}}\wedge\cdots\wedge e_{\iota_{(m,n_m-d^T(m))}} .\nonumber
\end{align}
We write
\begin{align*}
\alpha(T;\bd n;\bm\lambda) &=  \delta(\bm\lambda)\prod_{l=1}^m \frac{n_l!}{(n_l-d^T(l))!} \sum_{\substack{\sigma_\ell,\sigma_\ell'\in\Sigma, \ell\in T \\ \sigma'_{\iota}\in\Sigma,\iota\in\mc I}} \alpha'(T;\bd n;\bm\lambda;\bm\sigma)
\end{align*}
where $\bm\sigma $ stands collectively for all $\sigma,\sigma' $ indices in the sum and 
\begin{align}
\alpha'(T;\bd n;\bm\lambda;\bm\sigma) &= \sum_{\substack{\iota_{(1,1)},\ldots,\iota_{(m,n_m-d^T(m))}\in\mc I \\ \text{s.t. }\sigma_{\iota_{(l,k)}} = \sigma_{(l,k)}'\,\forall (l,k)\in\mc I   }}\; \prod_{\ell\in T}\hat C(\lambda_\ell(\bm\lambda,\bm\sigma,\bm\iota)) \nonumber\\\label{recurform}&\qquad\qquad  \times \prod_{l=1}^m   \hat w_{n_l}\big((\varsigma_{(\ell,l)} \lambda_{\ell}(\bm\lambda,\bm\sigma,\bm\iota))_{\ell\ni l},\lambda_{\iota_{(l,1)}},\ldots,\lambda_{\iota_{(l,n_l-d^T(l))}}\big) \\&\qqquad\qquad\qquad\times e_{\iota_{(1,1)}}\wedge\cdots\wedge e_{\iota_{(m,n_m-d^T(m))}}   \nonumber\\
\lambda_\ell(\bm\lambda,\bm\sigma,\bm\iota) &= \Big(\varsigma_\ell \sum_{l\in\mf o(\ell)}\sum_{k=1}^{n_l-d^T(l)}p_{\iota_{(l,k)}} ,\sigma_\ell,\sigma_\ell'\Big) \nonumber
\end{align}
To obtain this expression, we have inserted the definition (\ref{exprap}) of $\hat {\mc A}_{n(\bd n,m)}' $ into (\ref{basictopform}) and wrote the unrestricted sum over $\bm\iota = \big(\iota_{(1,1)},\ldots,\iota_{(m,n_m-d^T(m))}\big) $ as a sum over possible spin assignments $\sigma'_\iota $ to all legs $\iota\in \mc I $, times a restricted sum over all choices of $\iota_k $ compatible with such an assignment (remember that $\lambda_{\iota} = (p_{\iota},\sigma_{\iota}),\,\iota\in\mc I $ are fixed):
\begin{align*}
\sum_{\iota_{(1,1)},\ldots,\iota_{(m,n_m-d^T(m))}\in\mc I} f(\bm\lambda,\bm\iota) = \sum_{\sigma_\iota'\in \Sigma,\iota\in\mc I}\sum_{\substack{\iota_{(1,1)},\ldots,\iota_{(m,n_m-d^T(m))}\in\mc I \\ \text{s.t. }\sigma_{\iota_{(l,k)}} = \sigma_{(l,k)}'\,\forall (l,k)\in\mc I   }} f(\bm\lambda,\bm\iota).
\end{align*} 
For any such spin assignment $\bm \sigma $, $\alpha'(T;\bd n;\bm\lambda;\bm\sigma) $ can now be build recursively as follows\footnote{The recursive structure emerges only after this, somewhat artificial, doubling of the spin variables $\sigma_{\iota},\sigma_\iota' $.}: Define the linear map $\mc C_\ell(\bm\lambda,\bm\sigma):\mf G\to\mf G $ by
\begin{align*}
\mc C_\ell(\bm\lambda,\bm\sigma)(e_{\iota_1}\wedge\cdots\wedge e_{\iota_k}) &= \hat C_{\sigma_\ell,\sigma_\ell'}\big(p_{\iota_1}+\cdots+p_{\iota_k}\big) e_{\iota_1}\wedge\cdots\wedge e_{\iota_k}
\end{align*}
and, for $l\in\{1,\ldots,m \}$ a vertex of $T$, the multilinear map $\mc W_l(\bm\lambda,\bm\sigma):\prod_{l'\in\Pi^{-1}(l)} \mf G\times \mb C^{n_l-d^T(l)}\to\mf G $ by 
\begin{align*}
\mc W_l(\bm\lambda,\bm\sigma)&\Big(\big(e_{\iota^{l'}_1}\wedge\cdots\wedge e_{\iota^{l'}_{k_{l'}}}\big)_{l'\in\Pi^{-1}(l)};e_{\iota_1},\ldots,e_{\iota_{n_l-d^T(l)}}\Big)\\ &= \hat w_{n_l}\Big( \big(-\textstyle{\sum}_{l'\in\Pi^{-1}(l)}\textstyle{\sum}_{k=1}^{k_{l'}} p_{\iota^{l'}_k} - \textstyle{\sum}_{k=1}^{n_l-d^T(l)}p_{\iota_k},\sigma_{\{l,\Pi(l)\}}'\big),\\&\qquad\qquad\big( (\textstyle{\sum}_{k=1}^{k_{l'}}p_{\iota^{l'}_k} ,\sigma_{\{l,l'\}})\big)_{l'\in\Pi^{-1}(l)} ,(p_{\iota_1},\sigma'_{(l,1)}),\ldots,(p_{\iota_{n_l-d^T(l)}},\sigma'_{(l,n_l-d^T(l))})\Big)\\&\qqquad\times \bigwedge_{l'\in\Pi^{-1}(l)}e_{\iota^{l'}_1}\wedge\cdots\wedge e_{\iota^{l'}_{k_{l'}}} \wedge e_{\iota_1}\wedge\cdots\wedge e_{\iota_{n_l-d^T(l)}}.
\end{align*}
Set, for $\iota\in\mc I $,
\begin{align*}
\alpha_{\iota}(\bm\sigma) &= \sum_{\iota'\in\mc I\text{ s.t. }\sigma_{\iota'}=\sigma_{\iota}'}e_{\iota'}.
\end{align*}
Note that 
\begin{align*}
\sum_{\sigma_\iota'\in\Sigma}\alpha_{\iota}(\bm\sigma) &= \sum_{\iota'\in\mc I}e_{\iota'}.
\end{align*}
Define recursively
\begin{align}\label{eqrecursion}
\alpha'(l;\bd n;\bm\lambda;\bm\sigma) &= \mc C_{\{l,\Pi(l)\}}(\bm\lambda,\bm\sigma) \Big[\mc W_l(\bm\lambda,\bm\sigma)\Big(\big(\alpha'(l';\bd n;\bm\lambda;\bm\sigma)\big)_{l'\in\Pi^{-1}(l)} ;\\&\qqquad\qquad\qquad\qquad\alpha_{(l,1)}(\bm\sigma),\ldots,\alpha_{(l,n_l-d^T(l))}(\bm\sigma) \Big)\Big].\nonumber
\end{align}
Then $\alpha'(r;\bd n;\bm\lambda;\bm\sigma)  = \alpha'(T;\bd n;\bm\lambda;\bm\sigma)  $ for $r$ the root of $T$ (by convention, $\mc C_{\{r,\Pi(r)\}} = \id $).

\subsection{Bounds on the recursion}

Let $\Vert\cdot\Vert $ be any norm on $\mf G $ such that $\Vert e_{(1,1)}\wedge\cdots\wedge e_{(m,n_m-d^T(m))}\Vert = 1 $. We are interested in $\Vert \alpha'(T;\bd n;\bm\lambda;\bm\sigma)\Vert $, which would imply a bound on $\vert \hat{\mc A}_{n(\bd n,m)}(T;\bd n;\bm\lambda)\vert $. One might attempt to prove a bound recursively from (\ref{eqrecursion}). That is, one might look for constants $c_{l} $ and $d_l $ such that
\begin{align*}
\Vert\alpha'(l;\bd n;\bm\lambda;\bm\sigma) \Vert &= \Big\Vert\mc C_{\{l,\Pi(l)\}}(\bm\lambda,\bm\sigma) \Big[\mc W_l(\bm\lambda,\bm\sigma)\Big(\big(\alpha'(l';\bd n;\bm\lambda;\bm\sigma)\big)_{l'\in\Pi^{-1}(l)};\\&\qqquad\qquad\qquad\qquad \alpha_{(l,1)}(\bm\sigma),\ldots,\alpha_{(l,n_l-d^T(l))}(\bm\sigma)\Big)\Big] \Big\Vert \\
&\leq c_l  \;\Big\Vert \mc W_l(\bm\lambda,\bm\sigma)\Big(\big(\alpha'(l';\bd n;\bm\lambda;\bm\sigma)\big)_{l'\in\Pi^{-1}(l)} ; \alpha_{(l,1)}(\bm\sigma),\ldots,\alpha_{(l,n_l-d^T(l))}(\bm\sigma)\Big) \Big\Vert\\
&\leq c_l d_l \; \prod_{l'\in\Pi^{-1}(l)} \Vert \alpha'(l';\bd n;\bm\lambda;\bm\sigma)\Vert  \prod_{m=1}^{n_l-d^T(l)}  \Vert\alpha_{(l,m)}(\bm\sigma)\Vert 
\end{align*}
This would imply 
\begin{align*}
\Vert \alpha'(T;\bd n;\bm\lambda;\bm\sigma)\Vert &= \Vert \alpha'(r;\bd n;\bm\lambda;\bm\sigma)\Vert \leq \prod_{l=1}^m c_ld_l \prod_{\iota\in\mc I}\Vert \alpha_\iota(\bm\sigma)\Vert
\end{align*}
If, for example, had we chosen a norm $\Vert\cdot\Vert $ that equals an $\ell^p $ norm on one forms, we would get
\begin{align}\label{boundfundforms}
\sum_{ \sigma'_{\iota}\in\Sigma,\iota\in\mc I}\prod_{\iota\in\mc I}\Vert \alpha_\iota(\bm\sigma)\Vert&\leq \vert \Sigma\vert^{\frac{p-1}{p} n(\bd n,m)} \cdot n(\bd n,m)^{\frac 1p n(\bd n,m)}.
\end{align}
If the constants $c_l,d_l $ were uniform in $ n(\bd n,m) $, this would give a bound of the type (\ref{improbound}) as long as $p\geq 2 $.\\
Without further assumptions, the best possible constant $c_l $ is the operator norm of $\mc C_{\{l,\Pi(l)\}} $. Similarly, $d_l $ can be thought of as the norm of the multilinear map $\mc W_l $. In the simplified case where the kernel $\hat w_{n_l} $ factorizes as
\begin{align*}
\hat w_{n_l}\Big(\lambda_{\{l,\Pi(l)\}} ,&\big(\lambda_{\{l,l'\}}\big)_{l'\in\Pi^{-1}(l)} ,\lambda_1,\ldots,\lambda_{n_l-d^T(l)}\Big) \\&= w_{\{l,\Pi(l)\}}(\lambda_{\{l,\Pi(l)\}} ) \prod_{l'\in\Pi^{-1}(l)}w_{\{l,l'\}}(\lambda_{\{l,l'\}}) \cdot w_1(\lambda_1)\cdots w_{n_l-d^T(l)}(\lambda_{n_l-d^T(l)}),
\end{align*}
we have 
\begin{align*}
\mc W_l(\bm\lambda,\bm\varsigma)&\Big(\big(\alpha'(l';\bd n;\bm\lambda;\bm\sigma)\big)_{l'\in\Pi^{-1}(l)} ;\alpha_{(l,1)}(\bm\sigma),\ldots,\alpha_{(l,n_l-d^T(l))}(\bm\sigma) \Big)\\ &= \mc W_{\{l,\Pi(l)\}}\Big[\bigwedge_{l'\in\Pi^{-1}(l)} \mf W_{\{l,l'\}} \big[\alpha'(l';\bd n;\bm\lambda;\bm\sigma)\big]  \,\wedge \, \mc W_1\big[\alpha_{(l,1)}(\bm\sigma)\big]\wedge \\&\qqquad\qqquad\cdots\wedge \mc W_{n_l-d^T(l)}\big[\alpha_{(l,n_l-d^T(l))}(\bm\sigma) \big]\Big],
\end{align*}
where $ \mc W_{\ell} , \mc W_k ,\ell\ni l,k=1,\ldots,n_l-d^T(l),$ are multiplication operators on $\mf G $ defined similarly to $\mc C_\ell $. In this simplified case, we would therefore have
\begin{align*}
d_l\leq \Vert \mc W_{\{l,\Pi(l)\}}\Vert_{op} \prod_{l'\in\Pi^{-1}(l)} \Vert \mc W_{\{l,l'\}}\Vert_{op} \cdot \Vert \mc W_1\Vert_{op}\cdots \Vert \mc W_{n_l-d^T(l)}\Vert_{op} \cdot c_{\wedge},
\end{align*}
where $c_\wedge $ is the submultiplicativity constant of the wedge product,
\begin{align}\label{eqsubmult}
\Vert \alpha_1\wedge\cdots\wedge \alpha_n\Vert \leq c_\wedge \Vert \alpha_1\Vert \cdots \Vert \alpha_n\Vert.
\end{align}
For many natural norms, $c_\wedge $ is $n(\bd n,m) $ - dependent and so large that a recursive bound derived along the lines just explained is \emph{not} any more of the type (\ref{improbound}). In general, the inequality (\ref{eqsubmult}) seems to have gotten surprisingly little attention in the literature. We start with a simple discussion of this issue.

\subsubsection{The submultiplicativity constant of the wedge product}

We denote by $\Vert \cdot\Vert_p $ the $\ell^p $ norm on $\mf G $, i.e. if
\begin{align*}
\alpha &= \sum_{\iota_1<\cdots<\iota_k} \alpha(\iota_1,\ldots,\iota_k) e_{\iota_1}\wedge\cdots\wedge e_{\iota_k},
\end{align*}
then
\begin{align*}
\Vert \alpha\Vert_p &= \left[ \sum_{\iota_1<\cdots<\iota_k} \vert\alpha(\iota_1,\ldots,\iota_k)\vert ^p \right]^{\frac1p}.
\end{align*}
If
\begin{align*}
\alpha_m &= \sum_{\iota_1<\cdots<\iota_{k_m}} \alpha_m(\iota_1,\ldots,\iota_{k_m}) e_{\iota_1}\wedge\cdots\wedge e_{\iota_{k_m}},
\end{align*}
for $m=1,\ldots,n $, then
\begin{align*}
\alpha_1\wedge\cdots\wedge \alpha_n &= \sum_{\iota_1<\cdots<\iota_k} \alpha(\iota_1,\ldots,\iota_k) e_{\iota_1}\wedge\cdots\wedge e_{\iota_k}
\end{align*}
with $k=k_1+\cdots+k_n $ and
\begin{align*}
\alpha(\iota_1,\ldots,\iota_k)&= \sum_{\pi\in \mc S(k_1,\ldots,k_n)} \sgn\pi\, \alpha_1(\iota_{\pi(1)},\ldots,\iota_{\pi(k_1)})\cdots\alpha_n(\iota_{\pi(k_1+\cdots+k_{n-1}+1)},\ldots,\iota_{\pi(k_1+\cdots+k_{n})}),
\end{align*}
where $\mc S(k_1,\ldots,k_n) $ is the set of all permutations of $\{1,\ldots,k\} $ that preserve the order of the subsets $\{1,\ldots,k_1\},\ldots,\{k_1+\cdots+k_{n-1} ,\ldots, k_1+\cdots+k_{n}\} $ (shuffle permutations). Clearly, $\vert\mc S(k_1,\ldots,k_n)\vert = \binom{k}{k_1 \cdots\,k_n} $, and by destroying possible cancellations due to $\sgn\pi $, we obtain the bound
\begin{align}\label{simplesubmult}
\Vert \alpha_1\wedge\cdots\wedge\alpha_n\Vert_p \leq \binom{k}{k_1 \cdots\,k_n}^{\frac{p-1}{p}} \Vert \alpha_1\Vert_p\cdots \Vert \alpha_n\Vert_p.
\end{align}
This bound is not sharp, and more optimal estimates seem to be related to the poorly developed theory of tensor rank. If, for  instance, 
\begin{align*}
\alpha &= \alpha\ob 1\wedge\cdots\wedge \alpha\ob k,\;\alpha\ob j\in \mb C^{\mc I}
\end{align*}
has rank $1$, then some spectral properties of the operator $\delta_\alpha=\alpha\wedge :\mf G\to\mf G $ are easy to understand. Indeed, $\delta_\alpha^*\delta_\alpha $ is an orthogonal projection (modulo normalization) with respect to the $\ell^2 $ inner product, as can be seen inductively using the anticommutator
\begin{align}\label{eqccr}
\{\delta_{\alpha\ob j},\delta_{\alpha\ob j}^*\} &= \Vert \alpha\ob j\Vert_2^2,
\end{align}
or, equivalently, follows from the fact that $\delta_\alpha^*\delta_\alpha = \prod_{j=1}^k \Vert \tilde\alpha\ob j\Vert_2^2 n_j  $, where $\tilde\alpha\ob j $ are the orthogonal forms generated from $\alpha\ob j $ by the Gram-Schmidt algorithm, and $n_j $ is the fermionic number operator associated to $\tilde\alpha\ob j $, whose operator norm is bounded by $1$. In particular,
\begin{align*}
\Vert \alpha\wedge\alpha'\Vert_2\leq \Vert \alpha\Vert_2\cdot\Vert\alpha'\Vert_2
\end{align*}
if $\alpha $ has rank $1$. The anticommutation relations (\ref{eqccr}) could also be used to investigate the spectral radius of $\delta_\alpha $ for $ \alpha$ of higher rank. In the case of a two form $\alpha\in \mb C^{\mc I}\wedge \mb C^{\mc I} $, the rank decomposition is the spectral theorem for antisymmetric matrices, and for this special case, optimal constants in (\ref{eqsubmult}) have been obtained explicitly in \cite{iwaniec2004hadamard} (see also \cite{luo2014}). They grow slower than the ones of (\ref{simplesubmult}), but the investigation is too incomplete to be applied to the questions of this paper.

\subsubsection{Fourier transform and the standard bound}\label{standardbound}

The standard strategy to avoid submultiplicativity constants is to work with the $\ell^2 $ norm and decompose the form $\alpha $ of (\ref{basictopform}) into a sum of rank $1$ forms by reversing the Fourier transform. Indeed, if $\alpha = \alpha\ob1\wedge\cdots \wedge \alpha\ob k ,\alpha\ob j\in\mb C^{\mc I} $ has rank $1$, then 
\begin{align*}
\mc C_\ell(\bm\lambda,\bm\sigma) [\alpha ] &= \sum_{x\in\mb T} C_{\sigma_\ell,\sigma_\ell'}(x) \mc E(x;\bm\lambda)[\alpha\ob1]\wedge\cdots \wedge\mc E(x;\bm\lambda)[\alpha\ob k]
\end{align*}
where $\mc E(x;\bm\lambda):\mb C^{\mc I}\to\mb C^{\mc I} $ is the multiplication operator $\mc E(x;\bm\lambda)[e_\iota] = e^{ixp_\iota}e_\iota $. Similarly, if $\alpha_{\{l,l'\}} = \alpha_{\{l,l'\}}\ob1\wedge\cdots\wedge \alpha_{\{l,l'\}}\ob{k^{l'}} $, $l'\in\Pi^{-1}(l) $, all have rank $1$, then
\begin{align}\nonumber
\mc W_l(\bm\lambda,\bm\sigma)\Big(&\big(\alpha_{\{l,l'\}}\big)_{l'\in\Pi^{-1}(l)};\alpha_{(l,1)},\ldots,\alpha_{(l,n_l-d^T(l))} \Big)\\\nonumber &= \sum_{\substack{x^{l'}\in\mb T,l'\in\Pi^{-1}(l) \\ x_0,\ldots,x_{n_l-d^T(l)}\in\mb T  }} w_{n_l}\Big((x_0,\sigma'_{\{l,\Pi(l)\}}),\big((x^{l'},\sigma_{\{l,l'\}})\big)_{l'\in\Pi^{-1}(l)},\\[-20pt]&\qqquad\qquad\qquad(x_1,\sigma_{(l,1)}'),\ldots,(x_{n_l-d^T(l)},\sigma'_{(l,n_l-d^T(l))})\Big)\label{fouriersimple}\\[8pt]\nonumber& \qquad\qquad\times \bigwedge_{l'\in\Pi^{-1}(l)}\mc E(x^{l'}-x_0;\bm\lambda)[\alpha_{\{l,l'\}}\ob1]\wedge\cdots \wedge\mc E(x^{l'}-x_0;\bm\lambda)[\alpha_{\{l,l'\}}\ob k]\\\nonumber&\qquad\qquad\qquad\wedge  \mc E(x_1-x_0;\bm\lambda)[\alpha_{(l,1)}(\bm\sigma)]\wedge\cdots \wedge\mc E(x_{n_l-d^T(l)}-x_0;\bm\lambda)[ \alpha_{(l,n_l-d^T(l))}(\bm\sigma)].
\end{align}
Applying this to the induction, we see that
\begin{align*}
\alpha(T;\bd n;\bm\lambda) &=  \delta(\bm\lambda)\prod_{l=1}^m \frac{n_l!}{(n_l-d^T(l))!} \sum_{\substack{\sigma_\ell,\sigma_\ell'\in\Sigma, \ell\in T \\ \sigma'_{\iota}\in\Sigma,\iota\in\mc I}} \sum_{\substack{x_\ell\in\mb T,\ell\in T\\ x_\iota\in\mb T,\iota\in\mc I \\ x_l\in\mb T,l=1,\ldots,m}}\prod_{\ell\in T}C_{\sigma_\ell,\sigma_\ell'}(x_\ell) \\&\qquad\qquad\times\prod_{l=1}^m w_{n_l}\Big((x_l,\sigma'_{\{l,\Pi(l)\}}),\big((x^{\{l,l'\}},\sigma_{\{l,l'\}})\big)_{l'\in\Pi^{-1}(l)},\\[-10pt]&\qqquad\qquad\qquad(x_{(l,1)},\sigma_{(l,1)}'),\ldots,(x_{(l,n_l-d^T(l))},\sigma'_{(l,n_l-d^T(l))})\Big)\\
&\qquad\qquad\times \bigwedge_{\iota\in\mc I} \mc E(x\ob{\iota,T}(\bd x);\bm\lambda)[\alpha_\iota]
\end{align*}
where $\bd x $ stands for the collection of $x $'s in the sum and $x\ob{\iota,T}(\bd x)\in\mb T $ is a function of $\bd x $ depending only on $\iota $ and $T$. Since each $ \mc E(x\ob{\iota,T}(\bd x);\bm\lambda)[\alpha_\iota] $ is a one form and therefore of rank $1$, it follows from this that
\begin{align}\nonumber
\Vert \alpha(T;\bd n;\bm\lambda) \Vert_2 &\leq  \Vert C\Vert_1^{m-1}\prod_{l=1}^m \frac{n_l!}{(n_l-d^T(l))!} \Vert w_{n_l}\Vert_1 \cdot \sup_{\bd x}\prod_{\iota\in\mc I} \big\Vert\mc E(x\ob{\iota,T}(\bd x);\bm\lambda)[\alpha_\iota] \big\Vert_2\\
&\leq  \Vert C\Vert_1^{m-1}\prod_{l=1}^m d^T(l)!2^{n_l} \Vert w_{n_l}\Vert_1 \cdot n(\bd n,m)^{\half n(\bd n,m)}\label{fourierbound}
\end{align}
and therefore the standard bound 
\begin{align*}
\Vert\hat{\mc A}_n(T;\bd n) \Vert_\infty \leq\delta_{n=n(\bd n,m)} \cdot  \frac{ n(\bd n,m)^{\half n(\bd n,m)}}{n(\bd n,m)!} \cdot  \Vert  C\Vert_1^{m-1} \prod_{l=1}^m d^T(l)!2^{n_l}\Vert   w_{n_l}\Vert_1, 
\end{align*}
cf (\ref{improbound}). In some important applications, $\Vert C\Vert_1 $ is much larger than the factor $\Vert \hat C\Vert_\infty $ that was obtained earlier for the same bound without the $ \frac{n^{\frac n2}}{n!} $ improvement from antisymmetrization.

\subsubsection{The norm of multiplication operators}

The bound of the previous section can also be obtained by using on $\mf G $ the norm \cite{federer2014geometric}
\begin{align*}
\Vert \alpha\Vert &= \inf\Big\{\sum_{p=1}^r\Vert \alpha_p\Vert_2,\, \text{rank }\alpha_p = 1, \sum_{p=1}^r \alpha_p = \alpha \Big\}.
\end{align*}
This norm coincides with $\Vert\cdot\Vert_2 $ on $1 $ forms and $n(\bd n,m) $ - forms. Clearly, $\Vert \alpha\wedge\alpha'\Vert\leq \Vert\alpha\Vert\cdot\Vert\alpha'\Vert $ for any $\alpha,\alpha'\in\mf G $, and (\ref{fourierbound}) then follows from induction and the operator norm bounds
\begin{align*}
\sup_{\Vert\alpha\Vert\leq 1}\Vert \mc C_\ell(\bm\lambda,\bm\sigma)[\alpha]\Vert  &\leq \sum_{x\in\mb T} \vert C_{\sigma_\ell,\sigma_\ell'}(x)\vert\\
\sup_{\Vert \alpha_1\wedge \cdots\wedge \alpha_{n_l}\Vert\leq 1}\Big\Vert \mc W_l(\bm\lambda,\bm\sigma)\big(\alpha_1,\ldots,\alpha_{n_l}\big)\Big\Vert  &\leq \sum_{x_1,\ldots,x_{n_l}\in\mb T} \big\vert w_{n_l} \big((x_1,\sigma_{\{l,\Pi(l)\}}'),\ldots,(x_{n_l},\sigma'_{(l,n_l-d^T(l))})\big)\big\vert
\end{align*}
By way of contrast, the $\ell^p $ operator norms are compatible with the perturbative estimate (\ref{perestimate}):
\begin{align}
\sup_{\Vert\alpha\Vert_p\leq 1}\Vert \mc C_\ell(\bm\lambda,\bm\sigma)[\alpha]\Vert_p  &\leq \Vert \hat C\Vert_\infty \label{propbound}\\
\sup_{\Vert \alpha_1\wedge \cdots\wedge \alpha_{n_l}\Vert_p\leq 1}\Big\Vert \mc W_l(\bm\lambda,\bm\sigma)\big(\alpha_1,\ldots,\alpha_{n_l}\big)\Big\Vert_p  &\leq \Vert \hat w_{n_l}\Vert_\infty. \label{intbound}
\end{align}
Even for $ p=2$, these bounds can not be used in a simple induction, since the application of the latter yields
\begin{align*}
\Big\Vert\mc W_l(\bm\lambda,\bm\sigma)\Big(\big(&\alpha'(l';\bd n;\bm\lambda;\bm\sigma)\big)_{l'\in\Pi^{-1}(l)} ,  \alpha_{(l,1)}(\bm\sigma),\ldots,\alpha_{(l,n_l-d^T(l))}(\bm\sigma) \Big)\Big\Vert_2 \\&\leq  \Vert \hat w_{n_l}\Vert_\infty   \prod_{k=1}^{n_l-d^T(l)} \Vert \alpha_{(l,k)}(\bm\sigma)\Vert_2\cdot \Bigg\Vert\bigwedge_{l'\in\Pi^{-1}(l)} \alpha'(l';\bd n;\bm\lambda;\bm\sigma)  \Bigg\Vert_2,
\end{align*}
which can only be processed further by using on the last factor the generally insufficient bound (\ref{simplesubmult}) (unless further structure of the forms $ \alpha'(l';\bd n;\bm\lambda;\bm\sigma)  $ is used).

\subsubsection{Modifications to the recursive structure}

One way of thinking about the standard Fourier transform bound is as a modification to the recursive structure that exhibits at every step the new rank $1$ contributions. For example, if $\ell_1,\ldots,\ell_d $ are the branches at some vertex of $T$, $\alpha_1,\ldots,\alpha_d\in\mf G $, and $\alpha_c\ob j\in\mb C^{\mc I} $, $j=1,\ldots,k_c,c=1,\ldots,d $, are one forms, then
\begin{align*}
\bigwedge_{c=1}^d\mc C_{\ell_c}(\bm\lambda,\bm\sigma)\big[\alpha_c\wedge\alpha_c\ob1\wedge\cdots\wedge \alpha_c\ob{k_c}\big] &= \sum_{x_1,\ldots,x_d\in\mb T}\prod_{c=1}^dC_{\sigma_{\ell_c},\sigma_{\ell_c}'}(x_c)\\&\qquad\qquad\qquad\times\bigwedge_{c=1}^d \alpha_c'\wedge\mc E(x_c,\bm\lambda)[\alpha_c\ob1]\wedge\cdots\wedge \mc E(x_c,\bm\lambda)[\alpha_c\ob{k_c}]
\end{align*}
for modified $\alpha_c'\in\mf G  $. One might be interested in similar modification schemes that work directly in momentum space. As an example, for $d=2 $, consider the integration by parts identity\footnote{We assume here for simplicity that $\hat C $ is the restriction to $\mb T^* $ of a smooth compactly supported continuum expression.}
\begin{align}\label{modrec}
\mc C_{\ell}(\bm\lambda,\bm\sigma)[\alpha] \wedge \mc C_{\ell'}(\bm\lambda,\bm\sigma)[\alpha'] &= \int_{\mb R^D} \ud t\;\hat C_{\sigma_\ell,\sigma_\ell'}(t)\;\mc C_{\ell'}'(t;\bm\lambda,\bm\sigma)\big[\alpha \wedge  \mc X(t;\bm\lambda) [\alpha'] \big]\\&\qquad + \int_{\mb R^D} \ud t\;  \bm\nabla\hat C_{\sigma_\ell,\sigma_\ell'} (t)\; \mc C_{\ell'}(t;\bm\lambda,\bm\sigma)\big[\alpha \wedge  \mc X(t;\bm\lambda)[\alpha'] \big]\nonumber
\end{align}
where $\mc C_{\ell} (t;\bm\lambda,\bm\sigma),\mc C_{\ell}'(t;\bm\lambda,\bm\sigma) $ and $ \mc X(t;\bm\lambda) $ are linear with
\begin{align*}
\mc C_{\ell}(t;\bm\lambda,\bm\sigma)[e_{\iota_1}\wedge\cdots\wedge e_{\iota_k}] &= \hat C_{\sigma_\ell,\sigma_\ell'}(p_{\iota_1}+\cdots+p_{\iota_k}-t)e_{\iota_1}\wedge\cdots\wedge e_{\iota_k}\\
\mc C_{\ell}'(t;\bm\lambda,\bm\sigma)[e_{\iota_1}\wedge\cdots\wedge e_{\iota_k}] &= \bm\nabla\hat C_{\sigma_\ell,\sigma_\ell'}(p_{\iota_1}+\cdots+p_{\iota_k}-t)e_{\iota_1}\wedge\cdots\wedge e_{\iota_k}\\
\mc X(t;\bm\lambda)[e_{\iota_1}\wedge\cdots\wedge e_{\iota_k}] &=  \chi\big(t\leq p_{\iota_1}+\cdots+p_{\iota_k}\big)e_{\iota_1}\wedge\cdots\wedge e_{\iota_k}.
\end{align*}
Here, $\chi(t\leq t'),t,t'\in\mb R^D $ is the indicator function of $t_1\leq t'_1,\ldots, t_D\leq t'_D $ and $\bm\nabla = \pderi{t_1}\cdots\pderi{t_D} $. In particular, the $\ell^p $ operator norms of $\mc C_{\ell} (t;\bm\lambda,\bm\sigma),\mc C_{\ell}'(t;\bm\lambda,\bm\sigma) $ and $ \mc X(t;\bm\lambda) $ are bounded by $\Vert \hat C\Vert_\infty,\,\Vert\bm\nabla\hat C\Vert_\infty $ and $1$, respectively. If $\alpha $ were rank $1$, this would give
\begin{align}\label{modbound}
\big\Vert \mc C_{\ell}(\bm\lambda,\bm\sigma)[\alpha] \wedge \mc C_{\ell'}(\bm\lambda,\bm\sigma)[\alpha'] \big\Vert_2 \leq \big(\Vert \hat C\Vert_1\cdot \Vert \bm\nabla\hat C\Vert_\infty + \Vert  \bm\nabla\hat C\Vert_1\cdot \Vert\hat C\Vert_\infty\big) \Vert \alpha\Vert_2\cdot \Vert \alpha'\Vert_2
\end{align}
More general integration by parts schemes for higher degrees $d$ or deeper levels of recursion are possible, but shall not be investigated in this paper.

\begin{rem}

The above arguments have to be slightly modified for nontrivial interpolation parameters $s_\ell,\ell\in T $. Fix the $s_\ell $ and write $s^T = a^*a $ with $a\in \mb C^{m\times m} $ symmetric positive definite and $\sum_{l'} \vert a_{l,l'}\vert^2 = 1 $. The Grassmann Gaussian integral becomes
\begin{align*}
\int \psi^{l_1}(\xi_1)\wedge\cdots \wedge\psi^{l_{2n}}&(\xi_{2n})\ud\mu_{C\otimes s^T}(\psi^1,\ldots,\psi^m) \\&= \frac1{2^nn!}\sum_{\substack{ \pi\in\mc S_{2n} \\ k_1,\ldots,k_n \in\{1,\ldots,m\} }}\sgn\pi \prod_{i=1}^n C(\xi_{\pi(2i-1)},\xi_{\pi_{2i}}) \ol{ a}_{k_i,l_{\pi(2i-1)}} a_{k_i,l_{\pi(2i)}}
\end{align*}
and in the estimate (\ref{loopbound}), we replace
\begin{align*}
\int_{\mb L^*\times\Sigma} \prod_{m=1}^n&\ud\lambda_m \big\vert\hat{\mc A}_{2n}(T;\bd n;\lambda_1,-\lambda_1,\ldots,\lambda_n,-\lambda_n)\big\vert \\& \longrightarrow \sum_{k_1,\ldots,k_n\in\{1,\ldots,m\}}\int_{\mb L^*\times\Sigma} \prod_{m=1}^n\ud\lambda_m \big\vert\hat{\mc A}_{2n}(T;\bd n;\lambda_1,-\lambda_1,\ldots,\lambda_n,-\lambda_n,k_1,\ldots,k_n)\big\vert
\end{align*}
and need to bound 
\begin{align*}
\sum_{k_1,\ldots,k_n\in\{1,\ldots,m\}} \Vert \hat{\mc A}_{2n}(T;\bd n;\bd k)\Vert_\infty
\end{align*}
with the $\infty $ - norm with respect to the $\lambda $ dependence. Here, $ \hat{\mc A}_{2n}(T;\bd n;\bd k) $ is constructed just like $\hat{\mc A}_{2n}(T;\bd n) $ in section \ref{algrec}, but with $\alpha_\iota(\bm\sigma) $ there replaced by 
\begin{align*}
\alpha_\iota(\bm\sigma,\bd k)  &= \sum_{\iota'\in\mc I\text{ s.t. }\sigma_{\iota'}=\sigma'_\iota} a(\iota,\iota') e_{\iota'}\\
 a(\iota,\iota')&= \left\{\begin{array}{ll} \ol a_{k_il_{\iota'}} & \text{if }\iota\text{ is the }(2i-1)\text{th element of }\mc I\\ a_{k_il_{\iota'}} & \text{if }\iota\text{ is the }2i\text{th element of }\mc I\end{array}\right.
\end{align*}
The bounds go through unchanged except for the estimate on $\Vert \alpha_\iota(\bm\sigma,\bd k)\Vert $, where the bound
\begin{align*}
\sum_{\bm\sigma}\prod_{\iota\in\mc I} \Vert \alpha_\iota(\bm\sigma)\Vert_2\leq \vert\Sigma\vert^{\half n(\bd n,m)} n(\bd n,m)^{\half n(\bd n,m)}
\end{align*}
of (\ref{boundfundforms}) has to be replaced by
\begin{align*}
\sum_{\bd k,\bm\sigma}\prod_{\iota\in\mc I} \Vert \alpha_\iota(\bm\sigma,\bd k)\Vert_2&\leq \vert\Sigma\vert^{\half n(\bd n,m)} \sum_{\bd k}\prod_{\iota\in\mc I} \Big\Vert \sum_{\bm\sigma}\alpha_\iota(\bm\sigma,\bd k)\Big\Vert_2\\
&=  \vert\Sigma\vert^{\half n(\bd n,m)}\sum_{k_1,\ldots,k_{\half n(\bd n,m)} \in\{1,\ldots,m\}} \prod_{i=1}^{\half n(\bd n,m)} \Big\Vert \sum_{\iota\in \mc I}a_{k_il_\iota}e_\iota \Big\Vert_2^2\\
&=  \vert\Sigma\vert^{\half n(\bd n,m)} \left[\sum_{k=1}^m \sum_{\iota\in\mc I} \vert a_{kl_\iota}\vert^2\right]^{\half n(\bd n,m)} = \vert\Sigma\vert^{\half n(\bd n,m)} n(\bd n,m)^{\half n(\bd n,m)}
\end{align*}
which is the same as before.
\end{rem}

\section{Sample applications}\label{applications}

In this section, we apply the strategy described so far to derive two sample theorems about nonperturbative bounds on terms in the tree expansion corresponding to certain classes of trees. Generally, these classes are defined by restricting either the type or the number of branches. Other bounds similar to the ones stated, perhaps tailored towards a specific application, could also be derived, but even though we include a simple minded comparison between our new bounds and the standard ones in the case of single scale nonrelativistic Fermions, we will not venture here into any serious application of our ideas.

\subsection{Two classes of trees with improved bounds}

We discuss the simplest classes of trees that illustrate the use of the recursion (\ref{eqrecursion}), either using the submultiplicativity constant of (\ref{simplesubmult}) or the modification to the recursive structure of (\ref{modrec}).

\subsubsection{Trees with a bounded number of branches}

Our first bound uses only the momentum space norms of the perturbative bound, but is summable over the class of trees with a finite number of branches, i.e. such that $\sum_l \big[d^T(l)-2\big]\vee0\leq \text{ const} $.

\begin{thm}\label{thmsparse}
Let $T $ be a tree on $\{1,\ldots,m\} $ and $n_1,\ldots, n_m\geq 2 $. Let $A(T;n_1,\ldots,n_m) $ be the amplitude of the tree expansion (\ref{treeexp}) corresponding to $T$, with vertex $l\in\{1,\ldots,m\} $ having $n_l $ legs. Denote by $n=  \sum_l \frac{n_l}2-(m-1) $ the number of loop lines. Then,
\begin{align*}
\vert A(T;n_1,\ldots,n_m)\vert \leq \vert \mb T\vert \cdot \frac{n^n}{n!} \cdot \Vert \hat C\Vert_\infty^{m-1} \Vert \hat C\Vert_1^n \prod_{l=1}^m d^T(l)!2^{n_l}\big([d^T(l)-1]\vee1\big)^{n} \cdot \vert\Sigma\vert^{ n_l}\Vert \hat w_{n_l}\Vert_\infty.
\end{align*}
\end{thm}

\begin{proof}

$A(T;n_1,\ldots,n_m) $ is the integral over $s_\ell\in[0,1],\ell\in T $ of a function $A(T;\bd s;n_1,\ldots,n_m) $, see (\ref{treeexp}). Therefore $\vert A(T;n_1,\ldots,n_m)\vert \leq  \sup_{\bd s}\vert A(T;\bd s;n_1,\ldots,n_m)\vert $. We bound $\vert A(T;\bd s;n_1,\ldots,n_m)\vert $ at $\bd s = 1 $ only, the general case is similar. According to (\ref{loopbound}), 
\begin{align*}
\vert A(T;1;n_1,\ldots,n_m)\vert \leq \vert \mb T\vert \cdot \frac{(2n)!}{2^nn!}\cdot \Vert \hat C\Vert_1^n \Vert \hat {\mc A}_{2n}(T;\bd n)\Vert_\infty,
\end{align*}
where
\begin{align*}
\hat {\mc A}_{n}(T;\bd n;\lambda_1,\ldots,\lambda_n) &= \frac1{n!}\int \alpha(T;\bd n;\lambda_1,\ldots,\lambda_n)
\end{align*}
with $\alpha $ as in (\ref{basictopform}). Clearly,
\begin{align*}
\left\Vert \int \alpha(T;\bd n;\bm\lambda)\right\Vert_\infty \leq \prod_{l=1}^m d^T(l)!2^{n_l} \sum_{\bm\sigma}\left\Vert \int \alpha'(T;\bd n;\bm\lambda;\bm\sigma)\right\Vert_\infty
\end{align*}
with $\alpha' $ as in (\ref{recurform}). For any $\lambda_1,\ldots,\lambda_n\in\mb L^* $,
\begin{align*}
\left\vert\int \alpha'(T;\bd n;\bm\lambda;\bm\sigma)\right\vert &= \Vert  \alpha'(T;\bd n;\bm\lambda;\bm\sigma)\Vert_2
\end{align*}
for the Euclidean norm on forms, since $\alpha' $ is a top degree form. Regard an arbitrary leaf $r\in\{1,\ldots, m\}$ as the root of $T$. Starting with $r $, and moving downward in the rooted tree, we now employ to the recursion (\ref{eqrecursion}): 1. The bound (\ref{propbound}) that removes the operator $\mc C_\ell $; 2. The bound (\ref{intbound}) that removes the multilinear operator $ \mc W_l $, up to a wedge product of forms $\alpha'(l';\bd n;\bm\lambda;\bm\sigma) $ corresponding to vertices $l' $ in the next layer of the tree; 3. The bound (\ref{simplesubmult}) to remove the wedge product and close the recursion. The first step produces a factor $\Vert \hat C\Vert_\infty $ for every line $\ell\in T $. The second step produces a factor $\Vert\hat w_{n_l}\Vert_\infty\cdot \Vert\alpha_{(l,1)}(\bm\sigma)\Vert_2\cdots \Vert\alpha_{(l,n_l-d^T(l))}(\bm\sigma)\Vert_2\leq \Vert\hat w_{n_l}\Vert_\infty\cdot (2n)^{\half(n_l-d^T(l))} $ for every vertex $l\in\{1,\ldots,m\}$. The third step produces, for every vertex $l$, a factor $1 $ if $d^T(l) = 1,2 $ and a factor $$\frac{(\sum_{l'\in\Pi^{-1}(l)}k_{l'} )!}{\prod_{l'\in \Pi^{-1}(l)}k_{l'}!} \leq (d^T(l)-1)^{n} $$ if $d^T(l)\geq 3 $ and deg $\alpha'(l';\bd n;\bm\lambda;\bm\sigma) = k_{l'} $. Finally, the uncontrolled sum over $\bm\sigma $ gives a factor $\vert\Sigma\vert^{\sum_l n_l} $.

\end{proof}

\subsubsection{Trees with an arbitrary number of short branches}

Our second bound is for a particular tree with a large number of vertices of degree $3$, one of whose neighbors has to be a leaf, but uses Fourier transform on the interaction kernels associated to the leaves and the bound (\ref{modbound}) for pairs of propagator lines.

\begin{thm}\label{shortbranches}
Let $T_m$ be the tree on $\{1,\ldots,2m+2\}$ with edges $\{l,l+1\}$, $l=1,\ldots,m+1 $ and $\{l,l+m+1\} $, $l=2,\ldots,m+1 $. Let again $A(T_m;n_1,\ldots,n_{2m+2}) $ be the amplitude of the tree expansion (\ref{treeexp}) corresponding to $T_m$, with vertex $l\in\{1,\ldots,2m+2\} $ having $n_l $ legs, and denote by $n=  \sum_l \frac{n_l}2-(2m+1) $ the number of loop lines. Then,
\begin{align*}
\vert A(T_m;n_1,\ldots,n_{2m+2})\vert &\leq \vert \mb T\vert \cdot \frac{n^n}{n!} \cdot \Vert \hat C\Vert_\infty \cdot \mf c^{m}\cdot  \Vert \hat C\Vert_1^n \prod_{l=1}^{2m+2} d^{T_m}(l)!2^{n_l} \\&\qqquad\times \prod_{l=1}^{m+2} \vert\Sigma\vert^{ n_l}\Vert \hat w_{n_l}\Vert_\infty \cdot\prod_{l=m+3}^{2m+2} \Vert w_{n_l}\Vert_1.
\end{align*}
with
\begin{align*}
\mf c&= \Vert \hat C\Vert_1\cdot \Vert \bm\nabla\hat C\Vert_\infty + \Vert  \bm\nabla\hat C\Vert_1\cdot \Vert\hat C\Vert_\infty.
\end{align*}
as in (\ref{modbound}).
\end{thm}

\begin{proof}

We argue as in the proof of Theorem \ref{thmsparse}. We choose $1$ as the root of the tree, and apply to the recursion: 1. The bounds (\ref{intbound}) and (\ref{propbound}) that remove $ \mc W_1$ and $\mc C_{\{1,2\}} $, respectively; 2. The bound (\ref{intbound}) that removes $ \mc W_{l}$, $l=2,\ldots,m+1 $; 3. The bound (\ref{modbound}) that removes the operators $\mc C_{\{l,l+1\}},\mc C_{\{l,l+m+1\}} $, $l=2,\ldots,m+1 $; 4. The identity (\ref{fouriersimple}) that writes the forms $\alpha'(l;\bd n;\bm\lambda;\bm\sigma) $, $l=m+3,\ldots,2m+2 $ as a superposition of rank $1$ forms and removes these forms from the wedge product; 5. The bound (\ref{intbound}) that removes $ \mc W_{m+2}$. These steps produce the factors $\Vert \hat w_{n_1}\Vert_\infty\cdot \Vert \hat C\Vert_\infty $; $\Vert \hat w_{n_l}\Vert_\infty $, $l=2,\ldots,m+1 $; $\mf c^m $; $ \Vert w_{n_l}\Vert_1 $, $l=m+3,\ldots,2m+2 $; and $ \Vert \hat w_{n_{m+2}}\Vert_\infty$, respectively. There is an uncontrolled sum over the spin of any leg that was not Fourier transformed. This concludes the proof.

\end{proof}

\subsection{Single scale Fermi systems}

We give a short discussion of the previous bounds in the context of a single scale integration for the zero temperature Fermi gas. Strictly speaking, we discuss the analogous bounds that can be derived for the effective action
\begin{align*}
\Omega_C(W;\eta) &= \log\int e^{W(\psi+\eta)}\ud\mu_C(\psi),
\end{align*}
where $\eta = \big(\eta(\xi)\big)_{\xi\in \mb L} $ is a second copy of Grassmann generators that anticommute with all $\psi(\xi) $. Writing 
\begin{align*}
\Omega_C(W;\eta) &= \sum_{k\geq 0}\sum_{\substack{\xi_1,\ldots,\xi_k\in\mb L  \\ x\in \mb T }} A_k(\xi_1,\ldots,\xi_k)\eta(x+\xi_k)\wedge\cdots\wedge \eta(x+\xi_k)
\end{align*}
with antisymmetric coefficients $A_k(\xi_1,\ldots,\xi_k) $ (the connected amputated $k$ point function), the argument of section \ref{fermexp} gives tree expansions 
\begin{align}\label{expkpf}
A_k(\xi_1,\ldots,\xi_k)  &= \sum_{m\geq 1}  \frac1{m!}\sum_{T\text{ tree on }\ul m}  A_k(T;\xi_1,\ldots,\xi_k)
\end{align}
of which (\ref{treeexp}) is the special case $k=0 $ ($A_0 = \Omega_C(W)$). Using the supremum norm on the Fourier transform $\hat{A}_k(T;\lambda_1,\ldots,\lambda_k)$, all the previous bounds apply without major changes to these expansions. In particular, the analogue of the bound of Theorem \ref{thmsparse} is
\begin{align}\label{boundkpoint}
\Vert \hat A_k(T;n_1,\ldots,n_m)\Vert_\infty\leq \frac{n^n}{n!} \cdot \Vert \hat C\Vert_\infty^{m-1} \Vert \hat C\Vert_1^n \prod_{l=1}^m d^T(l)! 4^{n_l } ([d^T(l)-1]\vee 1)^n\cdot \vert \Sigma\vert^{n_l} \Vert \hat w_{n_l}\Vert_\infty,
\end{align}
where $A_k(T,n_1,\ldots,n_m) $ is the contribution to $A_k(T) $ with vertex $l\in \ul m $ having $n_l$ legs, and $n=\sum_l \frac{n_l}2-(m-1)-\frac k2 $ is the number of loop lines.\\ 
The single scale Fermi gas, as in section \ref{observables} (with $\mu=1=2m $), has
\begin{align}\label{propagator}
C\big((x^0, x;\sigma;\kappa),(x_0', x';\sigma';\kappa')\big) &= \Ant{}\;\delta_{\sigma,\sigma'} \delta_{\kappa,1}\delta_{\kappa',0}  \int_{\mb R^{d+1}} \frac{\ud p^0\ud^{d }p}{(2\pi)^{d+1}}\frac{\chi_j(p^0,p)e^{ip(x-x')} }{ip^0- p^2 +1}
\end{align}
with a cutoff function 
\begin{align*}
\chi_j(p^0,p) &= \phi\big(M^{2j}((p^0)^2+(p^2-1)^2)\big),
\end{align*}
where $M,j\gg1 $ and $\phi\in C_c^\infty([\half,2]) $. By the choice of the cutoff, $\hat C $ is supported in a region of volume const $ M^{-2j} $ and is bounded by const $ M^{ j} $. That is, $\Vert\hat C\Vert_\infty= \const M^j $ and $\Vert \hat C\Vert_1= \const M^{-j}  $.

\begin{cor}\label{corollary}
Let $A_k(T;\xi_1,\ldots,\xi_k), \,\xi_i = (x^0_i,x_i;\sigma_i;\kappa_i)\in \mb R^{d+1}\times \{\downarrow,\uparrow\}\times\{0,1\}, $ be the amplitude of the tree expansion (\ref{expkpf}) for the $k$ point function of the single scale Fermi gas with propagator (\ref{propagator}) and even, translation invariant interaction $W$ (the effective interaction at scale $j-1$). For $m,\mf b \in\mb N$, let 
\begin{align*}
T_{m,\mf b} &= \Big\{T\text{ tree on }\{1,\ldots,m\} \text{ s.t. }  \sum_{l=1}^m [d^T(l)-2]\vee 0\leq \mf b  \Big\}
\end{align*}
be the set of trees with at most $\mf b $ branches. Denote by $w_n(\xi_1,\ldots,\xi_n) $ the coefficient of $W$ of order $n$ in the fields (without momentum conserving delta function) and set
\begin{align}\label{modeffact}
w'_k(\xi_1,\ldots,\xi_k) &= \sum_{m\geq 1} \frac1{m!}\sum_{T\in T_{m,\mf b}}  A_k(T;\xi_1,\ldots,\xi_k)
\end{align}
Define the norm
\begin{align*}
\Vert w_n\Vert &=  \Vert \hat C\Vert_\infty \Vert \hat C\Vert_1^{\frac n2 -1} 8e^{-1}(4e^{\mf b+\half} )^n \Vert \hat w_n\Vert_\infty 
\end{align*}
Then
\begin{align*}
\Vert  w'_k\Vert \leq 8  (4e^{\mf b})^k  \frac{\Vert W\Vert}{1-\Vert W\Vert}
\end{align*}
where $\Vert W\Vert = \sum_n \Vert w_n\Vert $. In particular, if $\alpha\geq 1  $ is large enough, depending on $\mf b $, then 
\begin{align}
\Vert \hat w_n\Vert_\infty < M^{\frac{n-4}2j}\cdot \alpha^{-n}   \Longrightarrow  \Vert \hat w_k'\Vert_\infty\leq  M^{\frac{k-4}2(j+1)}\cdot\alpha^{-k}\cdot c ,\label{pertpc}
\end{align}
with $c$ independent of $j $.

\end{cor}

\begin{rem}
It is easy to prove that in position space $\Vert C\Vert_1\leq \const M^{dj} $, which can be sharpened to $\Vert C\Vert_1 = \const M^{\frac{d+1}{2}j} $ as discussed below. Using this and Fourier transform and the standard bound of section \ref{standardbound}, one obtains the classic statement \cite{feldman2002fermionic}
\begin{align}
\Vert \hat w_n\Vert_\infty< M^{\frac{n-d-3}2j}\cdot \alpha^{-n}   \Longrightarrow  \Vert \hat w_k'\Vert_\infty\leq  M^{\frac{k-d-3}2(j+1)}\cdot\alpha^{-k}\cdot c ,\label{nppc}
\end{align}
which holds for $w'_k =$ the sum over \emph{all} trees $T $ on $\ul m $ of $ A_k(T) $, but needs stronger assumptions on $w_n $ as $j\to\infty $. In particular, if $d\geq2$, in the language of the renormalization group, the four legged kernels $w_4 $ are marginal according to (\ref{pertpc}), but relevant according to (\ref{nppc}).

\end{rem}

\begin{proof}
The bound \ref{boundkpoint} and 
\begin{align*}
\prod_{l=1}^m ([d^T(l)-1]\vee 1)  = \exp\sum_{l=1}^m \log \big(1+[d^T(l)-2]\vee 0\big) \leq e^{\mf b}
\end{align*}
for all $T\in T_{m,\mf b} $ gives the claim
\begin{align*}
\Vert \hat w'_k\Vert_\infty&\leq  e^{1-\frac k2}\Vert \hat C\Vert_\infty^{-1}\Vert \hat C\Vert_1^{1-\frac k2} \sum_{m\geq 1} \frac{1}{m!}\sum_{T\in T_{m,\mf b}} \prod_{l=1}^m d^T(l)! \\&\qqquad\times \sum_{n_1,\ldots,n_m\geq 2}   \prod_{l=1}^m  \Vert \hat C\Vert_\infty \Vert \hat C\Vert_1^{\frac {n_l}2 -1}  e^{ -1} (4e^{\mf b + \half}\vert \Sigma\vert)^{n_l }  \Vert \hat w_{n_l}\Vert_\infty\\
&\leq    e^{1-\frac k2}\Vert \hat C\Vert_\infty^{-1}\Vert \hat C\Vert_1^{1-\frac k2}\sum_{m\geq 1}  \Big(\sum_{n\geq 2} \Vert \hat C\Vert_\infty \Vert \hat C\Vert_1^{\frac n2 -1} 8e^{-1}(4e^{\mf b+\half} )^n \Vert \hat w_n\Vert_\infty  \Big)^m,
\end{align*}
where in the last step we used (see, e.g., \cite{balaban2009power}, Lemma III.6)
\begin{align*} 
\sum_{T\text{ tree on }\ul m} \prod_{l=1}^m d^T(l)!\leq m!8^m.
\end{align*}

\end{proof}

\noindent
We do not discuss in detail the bounds on $\sum_{m\geq 1}\frac1{m!} A_k(T_m) $ that are a corollary to Theorem \ref{shortbranches}, since this would require some additional arguments concerning the mixed use of $\Vert\hat w_n\Vert_\infty $ and $\Vert w_n\Vert_1 $ in that theorem. Due to the fact that $\mf c = \const M^{(d+1)j} $, we recover the bound (\ref{nppc}) for that sum, i.e. the nonperturbative power counting implied by Theorem \ref{shortbranches} is the same as the classic result for the sum over all trees. However, we arrive at this statement without using sectorization of the Fermi surface, while the sharp bound $\Vert C\Vert_1 = \const M^{\frac{d+1}{2}j}  $ used in the derivation of (\ref{nppc}) does need this concept. Roughly speaking, sectorization of the Fermi surface means to write $\hat C $ as the sum of $\mf l^{-(d-1)} $ ($\mf l^{-1}\in\mb N $) rotated copies of 
\begin{align*}
\hat C_{sect} &=   \frac{\chi_j(p^0,p) }{ip^0- p^2 +1} \nu \Big(\big((p^1)^2+\cdots+(p^{d-1})^2\big)\cdot \mf l^{-2}\Big),
\end{align*}
for some $\nu\in C^\infty_c([0,1]) $. Each copy of $\hat C_{sect} $ fulfills better bounds than the $\mf l^{-(d-1)} $th fraction of the corresponding bound on $\hat C $ itself. For example,
\begin{align*}
\Vert \bm\nabla \hat C_{sect}\Vert_1&\leq \const  \frac{M^{dj}\mf l^{2d-2} }{1-\half  M^{-j}\mf l^{-2}}
\end{align*}
which is much smaller than $\mf l^{d-1} \Vert \bm\nabla\hat C\Vert_1 $ at the optimal $\mf l= M^{-\half j} $. Note that $\Vert \bm\nabla \hat C_{sect}\Vert_1 $ is still much bigger than $\mf l^{d-1} \Vert  \hat C\Vert_\infty $ at $\mf l= M^{-\half j} $, since cancellations between neighboring sectors are destroyed when summing up (copies of) $ \bm\nabla \hat C_{sect} $ to yield $\hat C $ via the fundamental theorem of calculus. Note that versions of the fundamental theorem of calculus
\begin{align*}
\hat C(p) &= \int \chi(p;p')\tilde{\bm\nabla} \hat C(p')\ud p'
\end{align*}
more adapted to the spherical geometry of the singularity give a bound $\Vert \tilde{\bm\nabla}\hat C\Vert_1 = \const M^j\sim \Vert\hat C\Vert_\infty $ even without sectorization, but in contrast to the standard version used in (\ref{modrec}), the operators $\tilde{\bm\nabla} $ will not be translation invariant any more.\\ It would be interesting to see how the ideas of sectorization could be used to improve the nonperturbative power counting of Theorem \ref{shortbranches} beyond the one of the classic result. \\[15pt]
\textbf{Acknowledgements:} I acknowledge the support of MIUR through the FIR grant 2013 ``Condensed Matter in Mathematical Physics (COND-MATH)'' (code RBFR13WAET). I would like to thank Alessandro Giuliani and the referees for useful comments and important corrections.

\bibliographystyle{plain}
\bibliography{pgram.bib}

\end{document}